\documentclass[sigconf, nonacm]{acmart}

\newcommand\vldbdoi{10.14778/3551793.3551867}
\newcommand\vldbpages{3249 - 3262}
\newcommand\vldbvolume{15}
\newcommand\vldbissue{11}
\newcommand\vldbyear{2022}
\newcommand\vldbauthors{\authors}
\newcommand\vldbtitle{\shorttitle} 
\newcommand\vldbavailabilityurl{https://github.com/antonis-m/HYDRA_VLDB}
\newcommand\vldbpagestyle{empty}

\usepackage[utf8]{inputenc}
\usepackage{epsfig,xspace,url,xspace,xcolor,multirow}
\usepackage[labelfont=bf,textfont={small, it}]{caption}
\usepackage{parskip}
\usepackage{tabu}
\usepackage{soul}
\usepackage{array}
\usepackage{subfigure}
\usepackage{todonotes}

\usepackage{amsmath}
\usepackage{amsthm}

\newtheorem{theorem}{Theorem}
\usepackage{tikz}
\usepackage{listings}
\usepackage{graphicx}
\usepackage{pifont}
\usepackage[linesnumbered, ruled, lined, noend]{algorithm2e}

\usepackage{enumitem}
\usepackage{tabularx}
\usepackage[table]{colortbl}
\usepackage{mathtools}
\def\multiset#1#2{\ensuremath{\left(\kern-.3em\left(\genfrac{}{}{0pt}{}{#1}{#2}\right)\kern-.3em\right)}}

\usepackage{mathtools}

\usepackage[a-1b]{pdfx}

\newcommand{\parentheses}[1]{\left(#1\right)}
\newcommand{\brackets}[1]{\left[#1\right]}
\newcommand{\set}[1]{\left\{#1\right\}}

\newcommand{\ie}{\emph{i.e.,}\xspace}
\newcommand{\eg}{\emph{e.g.,}\xspace}
\newcommand{\etc}{\emph{etc.}\xspace}
\newcommand{\etal}{\emph{et al.}\xspace}

\newcommand{\sysname}{{\sc Hydra}\xspace}

\newcommand{\antonis}[1]{{{#1}}}

\newcolumntype{C}[1]{>{\centering\let\newline\\\arraybackslash\hspace{0pt}}m{#1}}

\title{Enabling Efficient and General \Aggregation Analytics in  Multidimensional Data Streams}

\author{Antonis Manousis}
\affiliation{%
 \institution{Carnegie Mellon University}
}
\email{antonis@cmu.edu}

\author{Zhuo Cheng}
\affiliation{%
 \institution{Carnegie Mellon University}
}
\email{zhuoc2@andrew.cmu.edu}

\author{Ran Ben Basat}
\affiliation{%
 \institution{University College London}
}
\email{r.benbasat@cs.ucl.ac.uk}

\author{Zaoxing Liu}
\affiliation{%
 \institution{Boston University}
}
\email{zaoxing@bu.edu}

\author{Vyas Sekar}
\affiliation{%
 \institution{Carnegie Mellon University}
}
\email{vsekar@andrew.cmu.edu}

\newcounter{insightlabel}
\newcounter{insightnmbr}
\renewcommand{\theinsightlabel}{\textbf{\theinsightnmbr}}

\newcounter{packednmbr}

\newcommand{\mypara}[1]{\smallskip \noindent{\bf {#1}:}~}

\newcommand{\myparaq}[1]{\medskip\noindent{\bf {#1}?}~}

\newenvironment{packedenumerate}{\begin{list}{\thepackednmbr.}{\usecounter{packednmbr}\setlength{\itemsep}{0.5pt}\addtolength{\labelwidth}{-4pt}\setlength{\leftmargin}{\labelwidth}\setlength{\listparindent}{\parindent}\setlength{\parsep}{1pt}\setlength{\topsep}{0pt}}}{\end{list}}

\usepackage{xspace}
\newcommand{\aggregations}{subpopulations\xspace}
\newcommand{\aggregation}{subpopulation\xspace}
\newcommand{\Aggregation}{Subpopulation\xspace}

\newcommand{\subpopulations}{\aggregations}
\newcommand{\multidimensional}{multidimensional\xspace}
\newcommand{\analytics}{analytics\xspace}
\newcommand{\HYDRA}{{\sc Hydra}\xspace}
\newcommand{\sketch}{{\sc Hydra}-sketch\xspace}

\newcommand{\statistic}{summary statistic\xspace}
\newcommand{\statistics}{summary statistics\xspace}
\newcommand{\DataIngestion}{Data Ingestion\xspace}
\newcommand{\QueryEstimation}{Query Estimation\xspace}
\newcommand{\ingestion}{ingestion\xspace}

\newcommand{\datapoints}{data records\xspace}
\newcommand{\datapoint}{data record\xspace}

\newcommand{\dimensions}{dimensions\xspace}
\newcommand{\dimension}{dimension\xspace}

\newcommand{\datapointvar}{\ensuremath{\mathit{x}}}

\newcommand{\numdimensions}{\ensuremath{\mathit{D}}}
\newcommand{\dimensionval}{\ensuremath{\mathit{d}}}
\newcommand{\dimensionindex}{\ensuremath{\mathit{i}}}
\newcommand{\dimensionvar}{\ensuremath\mathit{D}}
\newcommand{\dimCardinality}{\ensuremath\mathit{C}}

\newcommand{\metricval}{\ensuremath{\mathit{m}}}
\newcommand{\metricname}{\ensuremath{\mathit{M}}}

\newcommand{\aggregationval}{\ensuremath{\mathit{Q}}}
\newcommand{\aggregationindex}{\ensuremath{\mathit{i}}}

\newcommand{\dimensionk}{\ensuremath{\mathit{l}}}

\newcommand{\query}{\ensuremath{\mathit{q}}}
\newcommand{\queryindex}{\ensuremath{\mathit{k}}}

\newcommand{\NumData}{\ensuremath{\mathit{V}}}
\newcommand{\NumDimensionsPerData}{\ensuremath{\mathit{D}}}

\newcommand{\NumAggregations}{\ensuremath{\mathit{Q}}}

\newcommand{\NumStatistics}{\ensuremath{\mathit{S}}}

\newcommand{\sketchwidth}{\ensuremath{\mathit{w}}}

 \newcommand{\sketchrows}{\ensuremath{\mathit{r}}}
 
  \newcommand{\Datastream}{\ensuremath{\mathbb{S}}}
 
\newcommand{\streamlength}{\ensuremath{\mathit{m}}}
 
\newcommand{\streamdistinct}{\ensuremath{\mathit{n}}}

 \newcommand{\frequency}{\ensuremath{\mathit{f}}}
 
  \newcommand{\gstatistic}{\ensuremath{\mathit{g}}}
  
   \newcommand{\Gstatistic}{\ensuremath{\mathit{G}}}
   
   \newcommand{\itemindex}{\ensuremath{\mathit{j}}}
   
  \newcommand{\CountSketch}{\ensuremath\mathit{CS}}
  
   \newcommand{\UnivSketch}{\ensuremath\mathit{US}}

   \newcommand{\hash}{\ensuremath{\mathit{h}}}
   
   \newcommand{\layerindex}{\ensuremath{\mathit{l}}}
   
    \newcommand{\rowindex}{\ensuremath{\mathit{k}}}

\newcommand{\NumLayers}{\ensuremath{\mathit{L}}}
\newcommand{\heavykeys}{\ensuremath{\mathit{k}}}

\newcommand{\conviva}{video-QoE\xspace}
\newcommand{\primary}{frontend\xspace}
\SetKwComment{Hline}{}{\vspace{-3mm}\textcolor{gray}{\hrule}\vspace{1mm}}

\begin{document}

\begin{abstract}
Today's large-scale services  (\eg video streaming platforms, data centers, sensor grids) need diverse real-time \statistics across multiple \aggregations of \multidimensional datasets. However, state-of-the-art frameworks do not offer general and accurate \analytics in real time at  reasonable costs. The root cause is the combinatorial explosion of data \aggregations and the diversity of \statistics we need to monitor simultaneously. We present \HYDRA, an efficient framework for multidimensional analytics that presents a novel combination of  using  a ``sketch of sketches''  to avoid the overhead of monitoring exponentially-many \aggregations and universal sketching to ensure accurate estimates for multiple statistics. We build \sysname as an Apache Spark plugin and address practical system challenges  to minimize overheads at   scale. Across  multiple real-world and synthetic multidimensional datasets, we show that \sysname can achieve robust error bounds and is an order of magnitude more efficient in terms of operational cost and memory footprint        than
existing frameworks (\eg Spark, Druid) while ensuring interactive estimation times.  
\end{abstract}
\maketitle

\pagestyle{\vldbpagestyle}
\begingroup\small\noindent\raggedright\textbf{PVLDB Reference Format:}\\
\vldbauthors. \vldbtitle. PVLDB, \vldbvolume(\vldbissue): \vldbpages, \vldbyear.\\
\href{https://doi.org/\vldbdoi}{doi:\vldbdoi}
\endgroup
\begingroup
\renewcommand\thefootnote{}\footnote{\noindent
This work is licensed under the Creative Commons BY-NC-ND 4.0 International License. Visit \url{https://creativecommons.org/licenses/by-nc-nd/4.0/} to view a copy of this license. For any use beyond those covered by this license, obtain permission by emailing \href{mailto:info@vldb.org}{info@vldb.org}. Copyright is held by the owner/author(s). Publication rights licensed to the VLDB Endowment. \\
\raggedright Proceedings of the VLDB Endowment, Vol. \vldbvolume, No. \vldbissue\ %
ISSN 2150-8097. \\
\href{https://doi.org/\vldbdoi}{doi:\vldbdoi} \\
}\addtocounter{footnote}{-1}\endgroup

\ifdefempty{\vldbavailabilityurl}{}{
\vspace{.3cm}
\begingroup\small\noindent\raggedright\textbf{PVLDB Artifact Availability:}\\
The source code, data, and/or other artifacts have been made available at \url{\vldbavailabilityurl}.
\endgroup
}

\section{Introduction}
\label{sec:introduction}


Many large-scale infrastructures (\eg Internet services, sensor farms, datacenter monitoring) produce \textit{multidimensional} data streams that are  growing both in data volume and dimensionality~\cite{macrobase, asta, gorilla, linkedin}. These multidimensional data contain  measurements of metrics along with metadata that describe said measurements across domain-specific dimensions. For instance,  video streaming services analyze user  experience issues     across dimensions, such as ISP, CDN, Device, City, \etc~\cite{junchenconext, CFA}. We see similar trends in other domains \eg network and  data center monitoring~\cite{aws1,aws2, univmon}.


In these settings, analysts  need \textit{interactive} and \textit{accurate} estimates of {\em diverse}  \statistics across \textit{multiple data \aggregations} of their data.  For instance,  video analysts want to monitor different statistics of viewer quality  across subpopulations of viewers (\eg entropy of bitrate in each major US city, \etc)~\cite{CFA}. 
 Similarly,  network  operators want  to analyze traffic  grouped by  combinations of their 5-tuple (srcIP,  dstIP, srcPort, dstPort, protocol)~\cite{univmon}.  \antonis{This is analogous to classical OLAP cube applications where  the number of cube vertices grows exponentially as more \aggregations are admitted.}


In such  {\em     multidimensional telemetry} settings, we ideally want frameworks offering high fidelity and interactive estimates at  low operational cost. However, there are two fundamental  challenges. First,  there is a combinatorial explosion of data subpopulations to monitor, which can result in exponential overhead in operational costs and resources. Second, estimating multiple statistics entails compute and/or memory overhead  proportional to the number of statistics of interest. 

We find that existing frameworks are fundamentally limited in terms of the tradeoff across operational cost, accuracy, and estimation latencies they can offer. Exact analytics frameworks (\eg  Spark~\cite{ spark},
Hive~\cite{hive}, Druid~\cite{druid}) that rely on horizontal resource scaling entail  
poor cost-performance tradeoffs as datasets become larger. While 
  approximate analytics~\cite{approximate} (\eg sampling- or    sketch-based 
analytics) can  trade off estimation accuracy for lower cost and improved interactivity, 
these too suffer  undesirable tradeoffs.  For instance, sampling-based approaches
provide generality across metrics and can handle many \aggregations, but  their accuracy guarantees can be weak. On the other hand,    sketch-based analytics (\eg~\cite{moment, opensketch,  
beaucoup, ting, ting2, coopstore, aqua, agarwal,ben2017constant,basat2020faster,basat2021salsa}) can offer robust accuracy guarantees, but  cannot address the combinatorial explosion of data \aggregations and also incur per-statistic effort.

\antonis{In this paper, we present \sysname, a  framework for efficient and general analytics over multidimensional data streams. \sysname builds on the novel combination of two key ideas.} 
First, to tackle the combinatorial explosion of \aggregations, we use  a ``sketch of sketches'' that enables memory efficient data stream summarization. This reduces the framework's data-resident memory footprint by one to two orders of magnitude compared to Spark- and
Druid-based alternatives and offers robust and provable accuracy guarantees. Second, to provide high-fidelity estimations simultaneously for many statistics, we leverage universal sketching~\cite{univmon}. Unlike canonical sketch-based approaches that deploy one custom sketch type per statistic~\cite{moment, ting, ting2}, a universal sketch estimates multiple different \statistics with only one sketching instance.  

\antonis{To the best of our knowledge, \HYDRA is the first work  to: (a) propose the  combination   of a sketch-of-sketches   with    universal sketching  for the multidimensional telemetry problem. While  some prior works have proposed the concept of a sketch-of-sketches, they do so for more narrow estimates of interest and do not demonstrate practical system implementations supporting a broad range of estimates (\eg~\cite{considine});  (b) analytically prove the theoretical guarantees of such a construction; and  (c) design a practical end-to-end system design and implementation of this idea using the theoretical analysis.}  
We build a  prototype \sysname on Apache Spark but note that  our core   design is  platform agnostic and can be ported to other streaming/batching systems as well~\cite{druid, hadoop, kafkastreams}.  We also implement  practical  optimizations to mitigate compute bottlenecks to further  reduce \sysname's  runtime and cost.

We evaluate \sysname using two real-world datasets; (1) a 2h-long, January 2019 CAIDA trace from the equinix-NYC vantage point~\cite{caida, caida2} and (2) an anonymized real-world trace of video QoE from a video analytics provider capturing the perceived QoE of viewers of a US-based content provider~\cite{conviva}.
To further evaluate the sensitivity of \sketch, we also leverage   a
synthetic multidimensional dataset drawn from a Zipf distribution with different parameter values~\cite{coopstore, ting}. 

We compare \sysname against six baselines: A native Spark-SQL implementation for exact analytics, a Spark-based  implementation that uniformly samples incoming data, a sketch-based approach that allocates one universal sketch instance per\aggregation, VerdictDB~\cite{verdictdb} (a sampling-based alternative) and two key-value based   implementations
(on Apache  Spark and Druid) that pre-aggregate data at ingestion time and provide precisely accurate analytics . 

Our evaluation shows that: (1) \sysname offers robust accuracy (mean error across statistics $\leq$5\% with 90\% probability) at 1/10 of the operational cost of exact analytics frameworks; (2) \sysname's configuration heuristics ensure close to optimal accuracy-memory tradeoffs; (3) \sysname's memory footprint scales sub-linearly with dataset size and number of data \aggregations. Combined with performance optimizations that improve end-to-end runtime by 45\%, \sysname offers  7-20$\times$ better query latency  than Spark- and Druid-based alternatives.
\section{Background and Motivation}
\label{sec:motivation}

In this section, we present several motivating scenarios, introduce key aspects of multidimensional telemetry, and discuss the limitations of existing analytics frameworks.

\subsection{Motivating Scenarios}

\mypara{Video Experience Monitoring} To maintain their  ad- and/or subscription-driven revenues,   video providers need to detect issues that can degrade viewer experience. To that end, analysts first collect  video session summaries (\ie per viewer measurements of video quality) and use them to \antonis{periodically (\eg every minute)} compute \statistics of various video quality metrics. This allows them to monitor viewer experience across multiple \textit{\subpopulations} of viewers~\cite{junchenconext,CFA,surus}. For instance, to track the entropy of bitrate and the L1 Norm of buffering ratio -- a common indicator of streaming anomalies -- for viewers in different cities, analysts may want to estimate the following query:

\begin{verbatim}
    SELECT City, Entropy(Bitrate), L1Norm(Buffering)
    FROM SessionSummaries
    GROUP BY City
\end{verbatim}

\mypara{Network Flow Monitoring}  Network operators commonly rely on control-plane  telemetry frameworks~\cite{gigascope, sonata} for tasks such as traffic engineering~\cite{deriving, univmon}, attack and anomaly detection~\cite{fastmonitoring} or forensics~\cite{worm}. These frameworks \antonis{periodically} monitor performance metrics (\eg flow distributions, per-flow packet sizes, latency, \etc) across different \aggregations of flows, \ie network flows grouped across combinations of packet header fields.  For instance, the operator might want to track indicators of DDoS attacks as follows:
\begin{verbatim}
    SELECT dstIP, Cardinality(srcIP)
    FROM FlowTrace
    GROUP BY dstIP
\end{verbatim}

These use cases share a problem structure that is characteristic of  \textit{multidimensional telemetry}. Queries that involve estimating many statistics across many data subpopulations appear in various settings, such as A/B testing~\cite{AB1, AB2}, exploratory data analysis~\cite{macrobase, exploration1}, operations monitoring~\cite{scuba}, and sensor deployments~\cite{joltik}. 

\subsection{ Requirements and Goals}

Drawing on these use cases,  we derive three key properties of  the telemetry problem we want to tackle: 

\begin{packedenumerate}
\item \textbf{Multidimensional Data}: We define a \multidimensional \datapoint as $\datapointvar = (\dimensionval_1, \dots,\ \dimensionval_\numdimensions,\ \metricval)$, where $\dimensionval_\dimensionindex$ is the value of a dimension $\dimensionvar_\dimensionindex$ and $\metricval$ is the value of metric $\metricname$.  In video, quality metrics  might be bitrate or buffering time whereas \dimensions might be the viewer's location, their player device, their ISP or CDN. Metrics and dimensions are domain- and  usecase-specific. 

\item \textbf{Analytics on Data Subpopulations}: Analytics are estimated in parallel across subpopulations of the input data. A \aggregation $\aggregationval_\aggregationindex$ is a collection of \datapoints $\{\datapointvar_i\}$ such that all  $\datapointvar_i \in \aggregationval_\aggregationindex$ match on a  subset of \dimension values. With a slight abuse of notation, we  define   $\aggregationval_\aggregationindex$ using this set of \dimension values, \ie $\aggregationval_\aggregationindex = \{\dimensionvar_{\aggregationindex,1} =  \dimensionval_{\aggregationindex,1} \land \dots \land \dimensionvar_{\aggregationindex,\dimensionk} = \dimensionval_{\aggregationindex,\dimensionk}\}$, where $\{\dimensionvar_{\aggregationindex,1},\dots, \dimensionvar_{\aggregationindex,\dimensionk}\} \subseteq \{\dimensionvar_1, \dots, \dimensionvar_\NumDimensionsPerData\}$; e.g.,   a data \aggregation  could be  NYC-based viewers using  AppleTV.

\item \textbf{Multiple statistics to estimate:} For each subpopulation, the operator  wants to estimate various \statistics \eg heavy hitters, entropy, cardinality, \etc A query $\query_\queryindex$ specifies a set of \aggregations $\{\aggregationval_\aggregationindex\}$ and a statistic $\gstatistic$ to estimate using the values $\metricval_\itemindex$ of $\datapointvar_\itemindex \in \aggregationval_\aggregationindex$. 

\end{packedenumerate}

In practice, operators have three requirements: (1) High fidelity for a broad set of statistics \ie robust, apriori configured, error bounds for as many statistics as possible; (2) Near real-time estimations and; (3) Low  footprint (e.g., cloud compute and memory costs).

\subsection{Prior Work and Limitations}

Prior work  has focused on developing two broad (but not mutually exclusive)  approaches for multidimensional telemetry.  The first enables distributed computations by horizontally scaling the frameworks' resources.  The second enables approximate analytics that sacrifice estimation accuracy for improved performance. 

\begin{packedenumerate}
 \item \mypara{Horizontal resource scaling} Here we find SQL and NoSQL analytics frameworks whose distributed design reduces estimation latency through horizontal scaling of server resources (\eg Spark~\cite{spark}, Hive~\cite{hive}, Hadoop~\cite{hadoop}, Dremel~\cite{dremel}, Druid~\cite{druid}, Flink~\cite{flink}). These frameworks scale their clusters with input data and can provide \textit{precisely exact} estimations. However, as data volume and dimensionality grow, i) deploying such clusters becomes increasingly expensive and ii) the continuous addition of resources eventually results in marginal estimation latency gains due to data shuffling overheads~\cite{blinkdb}. 
 \item \mypara{Approximate Analytics} Approximate analytics frameworks leverage sampling or data summarization algorithms in order to trade off accuracy for performance and cost. Sampling-based frameworks allow for low estimation latency by sampling data either online, at query time (the analyst applies the sampling operators and parameters as part of the estimation query)~\cite{QTspark, scope,  dremel, pig, hive} or offline, by means of a pre-processing step that creates data samples to be used at query time~\cite{blinkdb, congressional, aqua, sidirourgos}.

While there is a rich  body of  sampling-based efforts, these  have two key  shortcomings.  First, their accuracy guarantees are in the form of confidence bounds that are computed after query estimation has taken place and that depend on the statistic being estimated and the number of samples used~\cite{dbest}. Therefore, when an estimate does not meet accuracy requirements, frameworks  often fall back to using other samplers or precisely exact estimates. Second, to offset the resource overheads of producing offline samples, frameworks often make hard apriori  choices on what subsets of their data to create samples for (\eg BlinkDB~\cite{blinkdb} that mines query logs for frequently queried data or VerdictDB~\cite{verdictdb} that allows users to identify popular data tables).
In contrast, sketch-based analytics ensure bounded accuracy-memory  trade-offs for arbitrary workloads in sub-linear space~\cite{cormode2012synopses, count-min, loglog, univmon, opensketch}. These frameworks build compact data summaries at ingestion and use them to estimate statistics with apriori provable error bounds.
\end{packedenumerate}

We can also combine horizontal resource scaling and approximations. For instance, both Apache Spark and Druid allow for data summarization at ingestion time such that incoming data are stored as a key-value store where the keys are distinct $\langle\aggregationval_\aggregationindex,\ \metricval_j\rangle$ tuples and the values are their respective counts. These hybrid approaches enable data reduction without compromising the framework's ability to offer precise estimations. 

\mypara{Qualitative Analysis} Next, we analyze the   overhead to  process  multidimensional streams and the resident data  cost using the above solutions. Let us denote the  dataset size (in terms of the number of \datapoints) as $\NumData$ and let $Q$ be the number of data \aggregations. Assuming $\NumDimensionsPerData$ dimensions and that each \datapoint belongs in $2^\NumDimensionsPerData$ different \aggregations, then $Q=\mathcal{O}(2^\NumDimensionsPerData  \times \NumData)$. In practice, assuming each dimension has cardinality $\dimCardinality$, there are $\mathcal{O}(\dimCardinality^\NumDimensionsPerData)$ \aggregations in the dataset. In practice, we find that $\mathcal{O}(2^\NumDimensionsPerData  \times \NumData)$ is a tighter empirical bound for $Q$ and we will use that moving forward. In addition, given that the framework needs to  estimate $\mathcal{O}(\NumStatistics)$ different statistics, the number of \statistics to be estimated is an exponential $\mathcal{O}(\NumAggregations \times \NumStatistics) = \mathcal{O}(2^\NumDimensionsPerData \times \NumData \times \NumStatistics)$. Assuming, as it is the case for frameworks for precisely exact analytics, that the CPU and memory requirements for data ingestion and/or statistics estimation scale linearly with \aggregations,  we see that the framework's runtime, resource requirements and cost also scale exponentially.

\begin{figure}[tb]
        \includegraphics[width=1\linewidth]{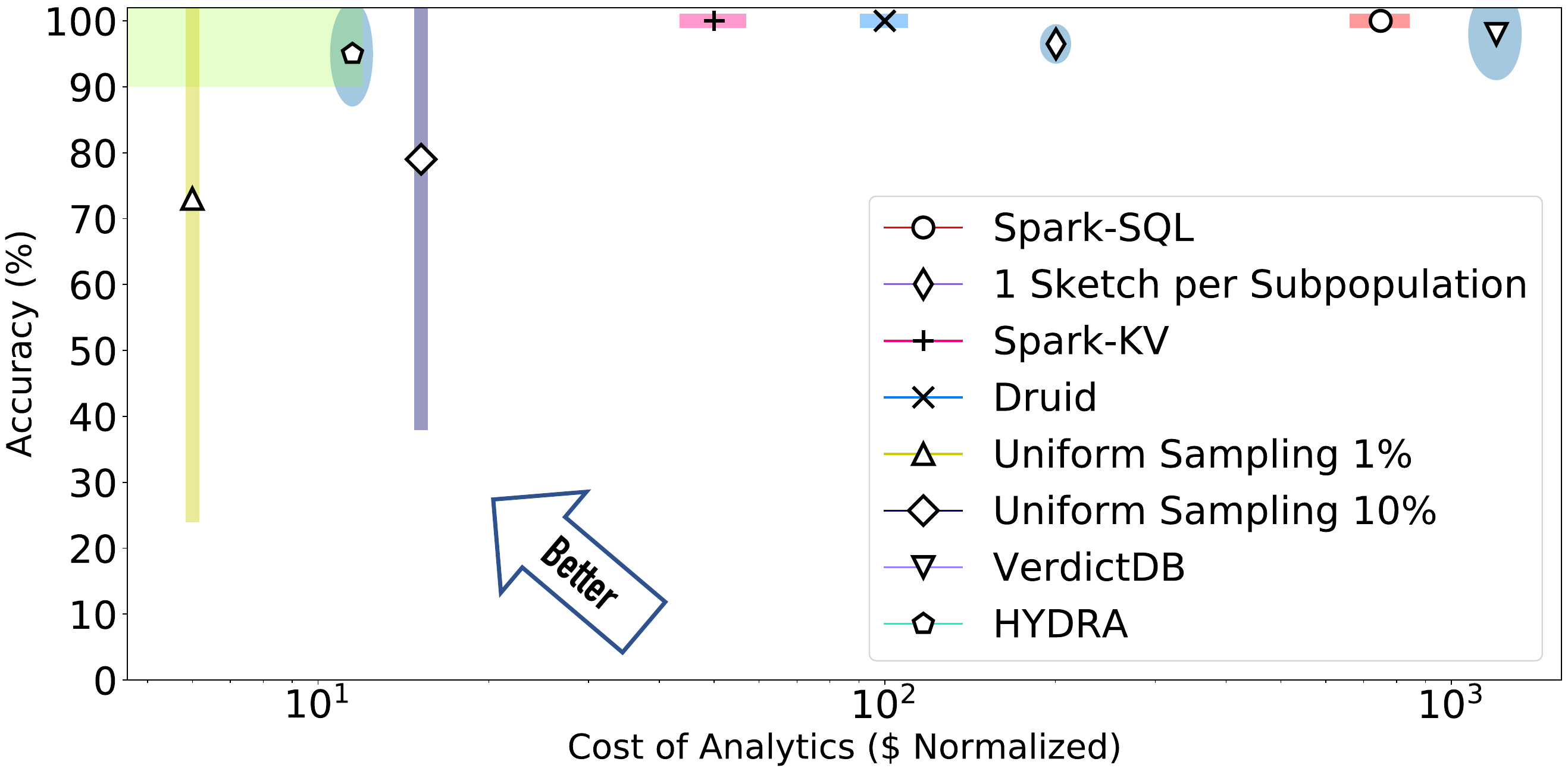}
        \caption{End-to-end cost of analytics. The green-shaded region indicates the ideal
        operating regime for \sysname.}
        \label{fig:sec7:cost}
\end{figure}

\mypara{Quantitative Analysis} To corroborate this qualitative analysis, we  evaluate the operational cost for several analytics frameworks when used in a multidimensional context (Figure~\ref{fig:sec7:cost}). Specifically, we measure their \$ cost as a function of their observed accuracy when asked to estimate in real-time 4 \statistics from a 130GB real-world dataset with approximately 5.6 million data \aggregations. Following the typical cloud billing model~\cite{pricing}, we use the total runtime $\times$ the number of cluster nodes used (20) as a proxy for the \$ cost. We  provide a detailed description of our experimental setup and baselines in \S\ref{sec:eval}.  Ideally, we need a framework whose cost-accuracy tradeoff lies in the  top-left, green region, \ie it offers the accuracy of a precise analytics framework  at the cost of sampling. However, we observe that the cost gap between the cheapest (1\% uniform sampling) and the most expensive baselines (precisely accurate Spark-SQL) is two orders  of magnitude wide. A sketch-based approach where the framework allocates one  sketch per \aggregation, while cheaper than Spark-SQL, remains expensive as it allocates exponentially many sketch instances, thus incurring high memory overheads. As discussed in \S\ref{sec:eval}, this baseline uses universal sketching that can simultaneously estimate all 4 statistics of interest per \aggregation with one sketch. Finally, precise baselines that summarize data at ingestion time, such as Apache Druid and Spark (denoted as Spark-KV) lie in the middle between Spark-SQL and sampling.

\mypara{Key takeaways}
 Multidimensional telemetry  entails a combinatorial explosion of data \subpopulations and \statistics to monitor. 
 Balancing cost, accuracy, and estimation latency is challenging due to the  combinatorial explosion in data \aggregations and the number of \statistics the framework needs to enable. Existing frameworks can only meet a subset of these goals, which motivates us to rethink how to support such analytics workloads at scale.


\section{\sysname: System Overview}
\label{sec:overview}

To support multidimensional workloads at scale, we envision \sysname as a streaming, sketch-based OLAP framework~\cite{druid, olapurl}. \sysname's distributed design (illustrated in Figure~\ref{fig:sec3:design}) includes one \primary and multiple worker nodes and its input are i) streams of multi-dimensional data, ingested in parallel at the worker nodes and ii) estimation queries provided by the  operator to the \primary node. \sysname implements two logical operations: \DataIngestion and \QueryEstimation. 

\begin{enumerate}[align=left, leftmargin=0pt, labelindent=\parindent, listparindent=\parindent, labelwidth=0pt, itemindent=!]
\item \textbf{\DataIngestion:} \DataIngestion happens at the worker nodes.  Each worker summarizes an incoming data stream to a local instance of \sketch. Data summarization happens on a per-\aggregation basis. Specifically, for every incoming \datapoint, \sysname first identifies what \aggregations the \datapoint belongs in and correspondingly updates a novel sketching primitive that we discuss below, \sketch. \sketch instances are configured to ensure accuracy guarantees and low  memory footprint (\S\ref{subsec:configs}).

\item \textbf{\QueryEstimation:} \QueryEstimation involves both \primary and worker nodes. The \primary receives operators' queries with i) the statistics to estimate \antonis{and ii) the set of \aggregations to estimate these statistics on}. 
Using this information, it creates a query plan that is distributed to the worker nodes who execute the queries. After estimation has taken place, the \primary node collects the results from the workers and returns them to the operator.

\end{enumerate}
\begin{figure}[tb]
        \includegraphics[width=1\linewidth, height=120pt]{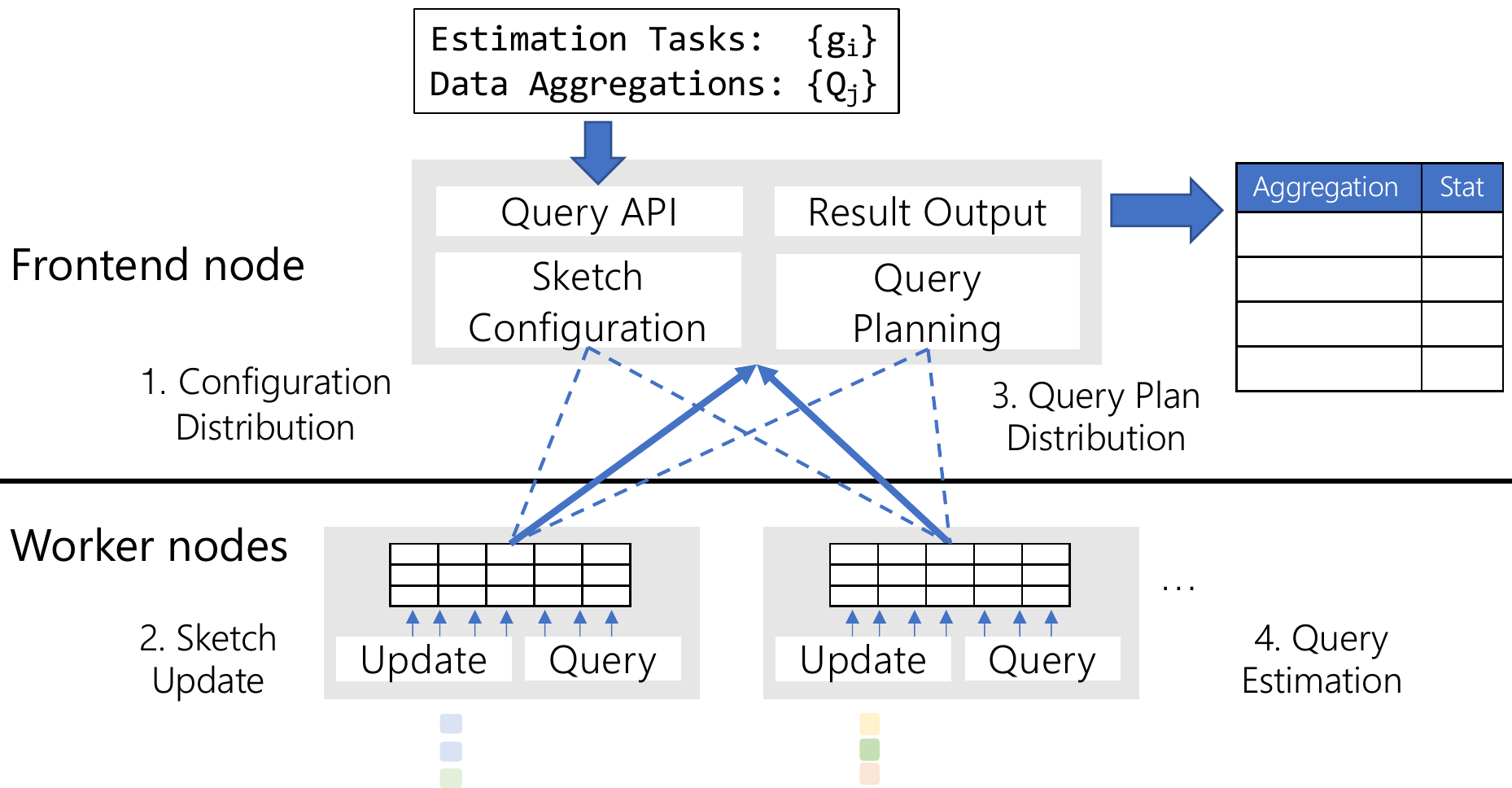}
        \caption{\sysname's example workflow. Workers perform data ingestion and querying. The \primary node exposes the query API to the operator and performs configuration and query plan dissemination.}
        \label{fig:sec3:design}
\end{figure}

\antonis{While the idea of using sketching to optimize analytics is not new, in our context canonical sketch-based approaches will need to  instantiate up to $\mathcal{O}(\NumStatistics)$ sketch instances per \aggregation. This is  inefficient as the framework needs exponentially many sketch instances, despite a sketch's ability to summarize a \aggregation's data in sub-linear space. }

\mypara{Key Idea}
To avoid the above limitations of conventional approaches, \HYDRA uses a novel combination of two 
 ideas. 

\begin{figure}[tb]
        \includegraphics[width=1\linewidth, height=120pt]{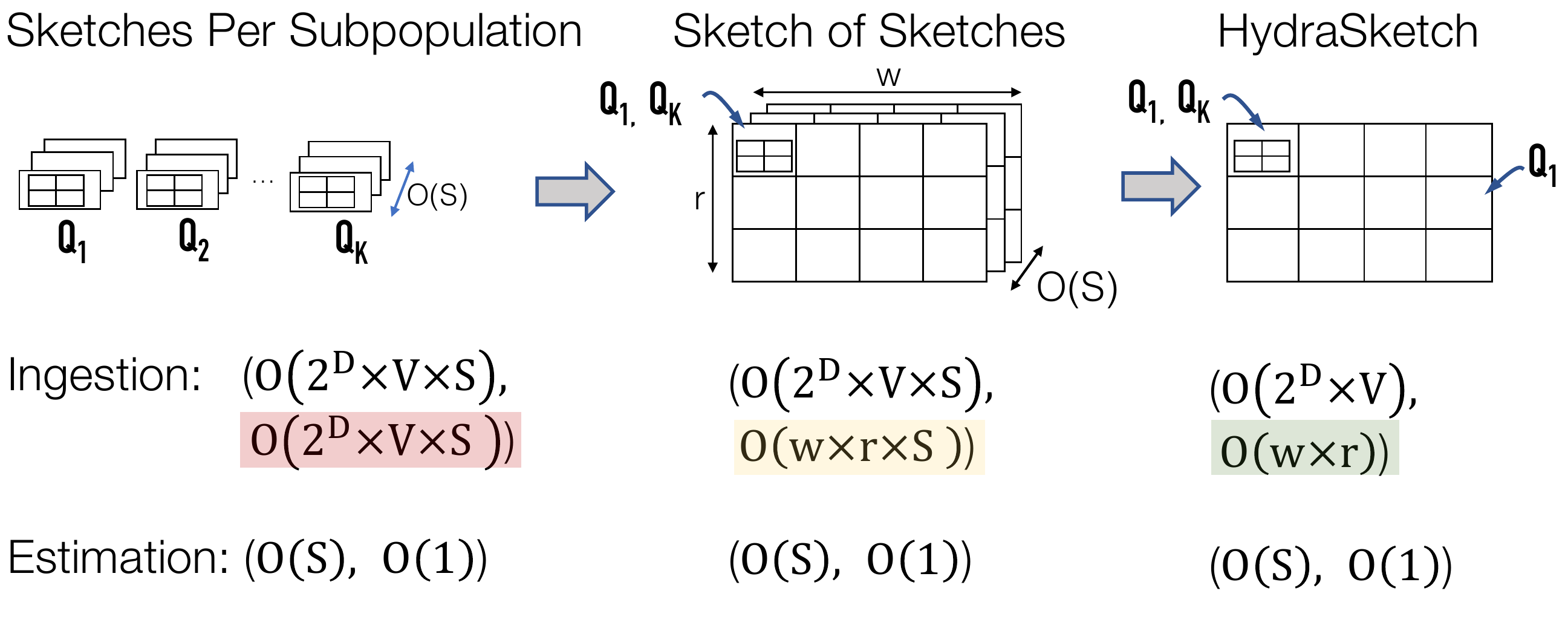}
        \caption{Comparison of Ingestion and Estimation (CPU time, space complexity) for different
        sketch-based designs. We highlight the theoretical improvements in space
        complexity from \sysname's design ideas.}
        \label{fig:insights}
\end{figure}

First, we observe that  we can  reduce  the exponential $\mathcal{O}(\aggregationval) = \ \mathcal{O}(2^\numdimensions \times \NumData)$ \ingestion-time, memory cost of sketch-based approaches through a novel ``sketch of sketches''.  We show that through a $\sketchwidth  \times \sketchrows$ array of sketch instances (Fig.~\ref{fig:insights}), where  $\sketchwidth \times \sketchrows \ll 2^\numdimensions \times \NumData$, \sysname reduces the memory cost of  estimating $\mathcal{O}(\NumStatistics)$ statistics from $\mathcal{O}(2^\numdimensions \times \NumData \times S)$ to $\mathcal{O}(\sketchwidth \times \sketchrows \times \NumStatistics)$. The intuition is that, unlike canonical sketch-based approaches, we can summarize multiple \aggregations into  one sketch instance and then query it with  predictable  error~\cite{considine,ting}.

Second, to reduce the need for instantiating $\mathcal{O}(\NumStatistics)$ different sketch types for $\mathcal{O}(\NumStatistics)$ \statistics,  \sysname leverages \textit{universal sketching}~\cite{univmon, zero_one_freq_law}.      Universal sketching enables replacing $\mathcal{O}(\NumStatistics)$ sketches with a single sketch \antonis{that simultaneously estimates multiple different statistics per subpopulation. This means that as long the desired statistics can be estimated with a universal sketch, there is no limit in the number of statistics that the sketch can estimate with a fixed memory footprint.} This design choice further reduces the framework's space complexity from $\mathcal{O}(\sketchwidth\times \sketchrows
\times \NumStatistics)$ to $\mathcal{O}(\sketchwidth\times \sketchrows)$. 

\antonis{While these two ideas (sketch of sketches and universal sketching)  have been independently proposed in other narrower contexts, to the best of our knowledge, we are the first effort to: (1) propose the combination of these ideas to tackle the multidimensional telemetry problem; (2) rigorously prove the accuracy-resource tradeoffs of this construction; and (3)  demonstrate a practical end-to-end realization atop state-of-art horizontally scalable ``BigData'' platforms. }

\section{\sysname Detailed Design}  
\label{sec:sketch}

We first provide background on sketching  to set up the intuition for \sketch. We then introduce the basic \sketch algorithm, formally prove its error bounds, and devise \sketch configuration strategies. Table~\ref{tab:notation} summarizes the notation we use.

\begin{table}[tb]
\centering
\caption{\sysname Notation. The upper subsection introduces notation specific to the sketch-of-sketches and the lower to universal sketches.}
\begin{small}
    \begin{tabular}{cl} 
    \toprule
    \textbf{Notation} & \textbf{Definition}  \\ 
    $\NumData$ & Input size \\
    $\NumDimensionsPerData$ & Number of data dimensions \\
    $Q$        & Number of data \aggregations \\
    $\NumStatistics$ & Number of \statistics \\
    $\Datastream_{\streamlength,\streamdistinct}$ & Stream of length $\streamlength$ and $\streamdistinct$ distinct keys \\
    \midrule
     $\sketchwidth$ & Number of sketches per 2D-sketch row \\
     $\sketchrows$  & Number of rows in 2D-sketch  \\
     $(\epsilon, \delta)$  & $0<\epsilon<1$ as additive error and $\delta$ is the probability that \\
     & the result error is not bounded by $\epsilon$ (failure probability)\ \\
     \midrule 
     $\sketchwidth_{\UnivSketch}$ & Number of counters per universal sketch row \\
     $\sketchrows_{\UnivSketch}$  & Number of universal sketch rows \\
     $(\epsilon_{\UnivSketch}, \delta_{\UnivSketch})$ & $0<\epsilon_{\UnivSketch}<1$ as additive error in universal sketch and\\
     &$\delta_{\UnivSketch}$ is the failure probability \\
     $\NumLayers$                 & Number of universal sketch layers \\
     $\heavykeys$                 & Number of keys in universal sketch heavy hitter heaps \\
     
    \bottomrule
    \end{tabular}
\end{small}

\label{tab:notation}
\end{table}

\subsection{Background on Sketching} 
\label{subsec:backgroundsketch}
Let $\Datastream_{\streamlength,\streamdistinct}$ denote a data stream with length $\streamlength$ and $\streamdistinct$ distinct keys. Suppose we want to estimate a frequency-based \statistic of the keys (\eg entropy, cardinality, frequency moments). A natural design is to estimate the desired statistic with a key-value  data structure tracking the frequency per key. For instance, for frequency estimation, we can maintain and increment one counter per key. While correct, the  space complexity is linear in $\streamdistinct$ and not space efficient (Figure~\ref{ins1}).
\begin{figure}[h]
        \centering
        \includegraphics[width=1\linewidth]{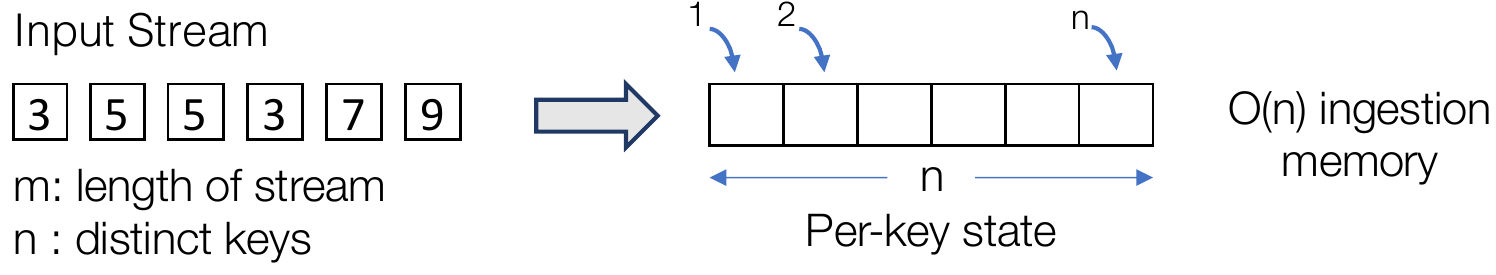}
        \caption{Maintaining per-key state is not space efficient}
        \label{ins1}
\end{figure}

\mypara{Hash-based mappings for space efficiency} To ensure sub-linear (in $\streamdistinct$) space complexity, sketching algorithms do not maintain per-key state but, instead, map multiple keys to the same counters via  hashing. For instance, a simple sketch for frequency estimation   consists of $\sketchwidth$ integer counters, where $\sketchwidth \ll \streamdistinct$.  Based on the hash of the key, an element gets mapped to a counter, which is then incremented to maintain an estimate of that key's frequency. Naturally,  multiple keys {\em colliding}  introduces some  errors (Figure~\ref{ins2}). 
\begin{figure}[h]
        \centering
        \includegraphics[width=1\linewidth]{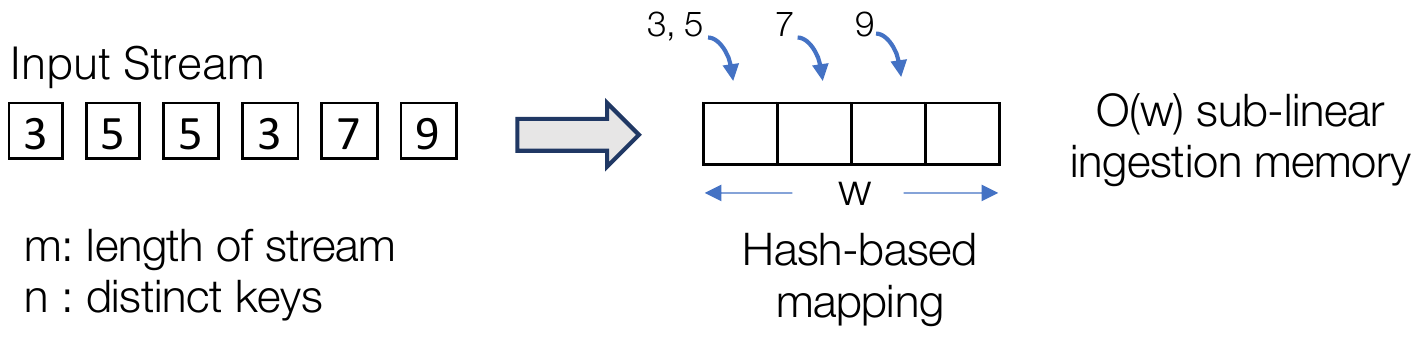}
        \caption{Hashing enables sub-linear memory complexity}
        \label{ins2}
\end{figure}

\mypara{Multiple independent updates for tighter error bounds} As defined, this basic mechanism only provides a small probability that the estimation error will lie within a desirable range of error values~\cite{ams}.  To overcome this, sketches use independent instances (\eg $\sketchrows$ arrays) of the counter structure of length $\sketchwidth$. Each vector  of length $\sketchwidth$ has its own hash function and the $\sketchwidth$ hash  functions are pairwise independent. Thus, ingesting a stream element now  translates to $\sketchrows$ update operations (\eg incrementing $\sketchrows$ integer counters instead of one). For each key, this sketch produces $\sketchrows$ different estimates of the statistic of interest. The final estimate will be a summary of $\sketchrows$ estimates (\ie min, median \etc) (Figure~\ref{ins3})~\cite{cormode2012synopses}. This amplifies  the probability that the estimation error lies within the desired range.
\begin{figure}[h]
        \centering
        \includegraphics[width=1\linewidth]{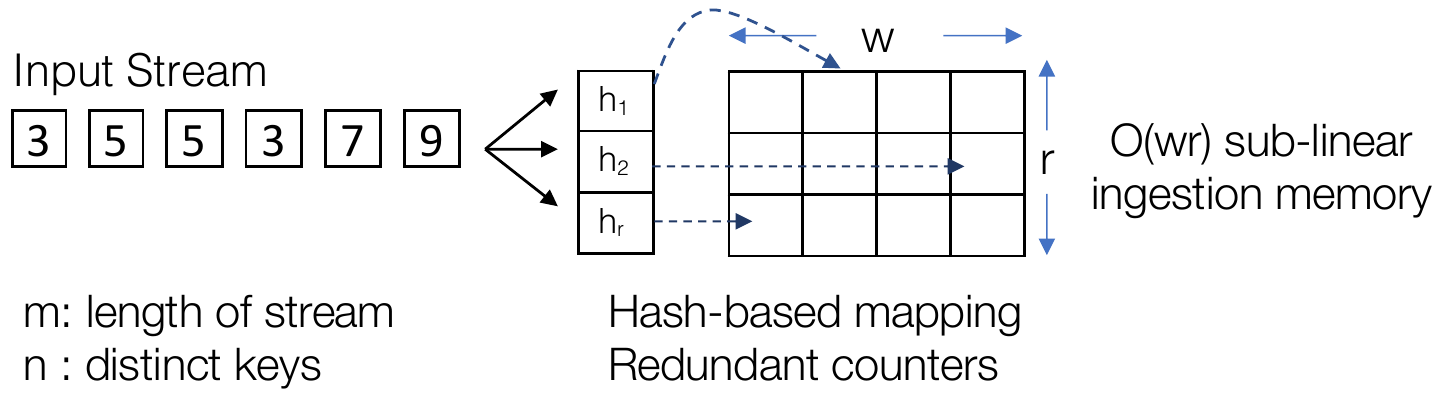}
        \caption{Independent hashing improves accuracy.}
        \label{ins3}
\end{figure}

\subsection{Tackling Subpopulation Explosion}

For now, let us make the simplifying assumption (which we relax later) that
our system only needs to estimate one \statistic (\eg entropy) per     data
\aggregation. 
Similar to Figure~\ref{ins1}, a starting point for our    design
would be to maintain per-\aggregation state, \ie allocate one sketch instance
for each of the $\mathcal{O}(2^\numdimensions \times \NumData)$ distinct  \aggregations.
This approach\antonis{, similar to an OLAP cube,} is not scalable as it requires as many sketches as  the number of data \aggregations.

To avoid keeping per-\aggregation state, we borrow from the   first intuition that we saw in the sketch construction in the background (Figure~\ref{ins2}). The basic sketch construction avoids maintaining per-key state by allowing multiple keys to explicitly collide in a hashed key-value store whose size is less than the number of unique elements. 

Note that the basic sketch is maintaining a single counter per array entry but we want to be able to estimate  some statistical summary of a \aggregation instead. Therefore, instead of keeping a single counter per array entry, we maintain a sketch-per-entry. This brings us to the following construction (Figure~\ref{ins4}). We consider a
single array of $\sketchwidth$ (\eg $\sketchwidth \ll 2^\numdimensions \times  \NumData$) \emph{sketches}. For each ($\aggregationval_i$, $\metricval_j$) pair, we hash the $\aggregationval_i$ and map it to one of the $\sketchwidth$ sketches, thus colliding multiple subpopulations to the same sketch. Then, we update the sketch with $\metricval_j$ and at query time, we estimate the statistic for $\aggregationval_i$.

\begin{figure}[h]
        \centering
        \includegraphics[width=1\linewidth]{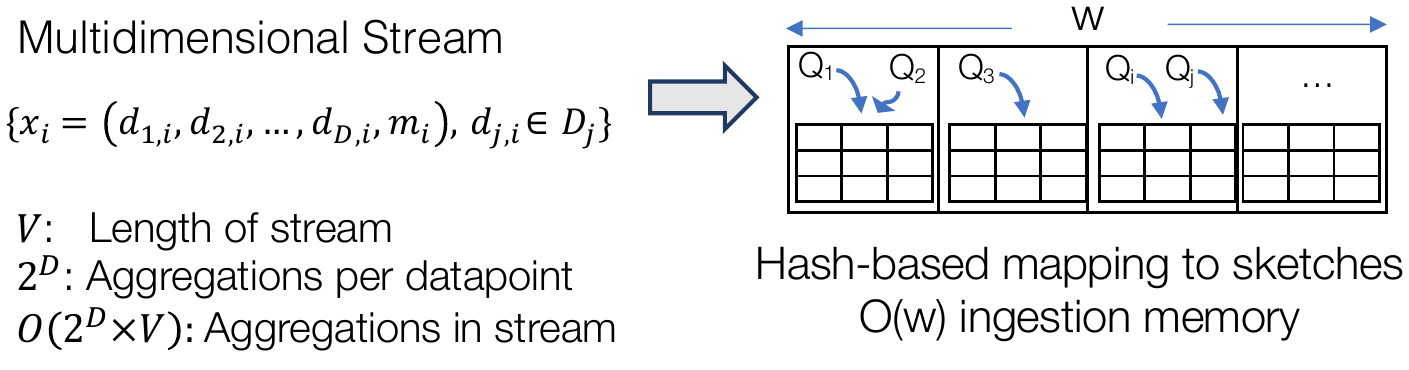}
        \caption{Hash-based mapping of \aggregations to a sketch vector.}
        \label{ins4}
\end{figure}

Analogous to the basic sketch from \S\ref{subsec:backgroundsketch}, by mapping multiple \aggregations to one sketch, this baseline construction will have some estimation error. To control this, we extend the idea of using redundant counter vectors and pairwise-independent hashes (Figure~\ref{ins3}). 
That is, we use $\sketchrows$ \emph{arrays of $\sketchwidth$
sketches} and use $\sketchrows$ pairwise-independent hash functions to map each \aggregation to one sketch per row (Figure~\ref{ins5}). At query time, we return the median of the $\sketchrows$ estimates. 
\begin{figure}[h]
        \centering
        \includegraphics[width=1\linewidth]{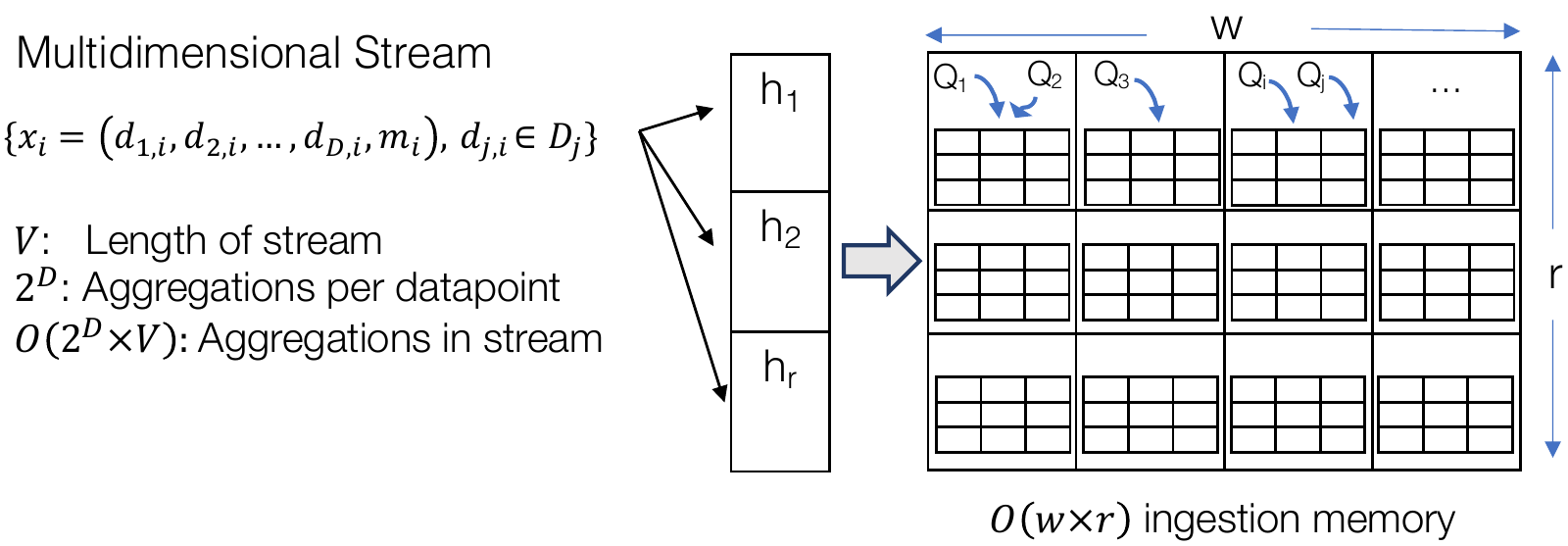}
        \caption{Redundant sketch vectors and pairwise-independent hashes for tighter error bounds.}
        \label{ins5}
\end{figure}

In summary, the above sketch-of-sketch construction maintains a {\em 2D array of sketches}  to  track multiple 
\aggregations. This  reduces  the memory cost of ingestion to 
$O(\sketchwidth\times\sketchrows)$ \ie sub-linear in \aggregations.    
In \S\ref{subsec:proof}, we formally prove the memory-accuracy 
tradeoffs for this construction.

\subsection{ Enabling Multiple Statistics}

The above discussion is based on the simplifying assumption that we
need to only estimate one \statistic. Since sketching algorithms      are
generally custom designed per statistic,  to  support  $O(\NumStatistics)$ different
\statistics,   we need to create  $O(\NumStatistics)$ sketch-of-sketches instances.
This  raises two natural concerns. 
First, the total memory cost of this solution becomes $O(\sketchwidth\times\sketchrows\times\NumStatistics)$,    \ie
linear to the number of \statistics of interest. Second, the framework     cannot
offer generality as it cannot estimate \statistics that are not already allocated; \eg
some future analysis might require estimating the entropy of a metric 
but the framework has not instantiated an entropy-specific sketch-of-sketch instance.

Our insight here is that the sketch of sketches structure can be combined    with
universal sketching~\cite{univmon}  to achieve the desired generality across
statistics. A universal sketch is a sketching primitive that enables    the
simultaneous estimation of multiple different, apriori unknown,   statistics
with one sketch instance. Therefore, instead of  a sketch-of-sketches
per statistic, we can use  one sketch of universal sketches. We formally
prove this in \S\ref{subsec:proof} and 
show that \sysname's ingestion cost drops to $\mathcal{O}(\sketchwidth \times \sketchrows)$.

\mypara{Background on universal sketches} A universal sketch can estimate any \statistic that belongs to a broad class of functions, known as {\em Stream-PolyLog}~\cite{braverman2014universal, zero_one_freq_law, univmon}. We denote each function in {\em Stream-PolyLog}, as $\Gstatistic\text{-sum} = \sum{\gstatistic(\frequency_\itemindex)}$, where $\frequency_\itemindex$ is the frequency of the $\itemindex$-th unique element in the input stream $\Datastream_{\streamlength, \streamdistinct}$~and $\gstatistic$ is a function defined over $\frequency_\itemindex$. If $\gstatistic$ is monotonically increasing and upper bounded by $\mathcal{O}({\frequency_\itemindex}^2)$, then $\Gstatistic\text{-sum}$ can be computed by a single universal sketch with polylogarithmic memory. \antonis{Universal sketch provides $\epsilon$-additive error guarantees to {\em Stream-PolyLog} and demonstrates better memory-accuracy tradeoffs than the composition of custom sketches when estimating multiple statistics from {\em Stream-PolyLog} in practice~\cite{univmon}. Key statistics of interest can be formulated via a suitable $\Gstatistic$-sum $\in$ Stream-PolyLog. 
Such examples include: $\alpha$-Heavy Hitters ($f_i \geq \alpha\Sigma f_i$),  L1-Norm  ($\Sigma f_i$), L2-Norm ($\Sigma f_i^2$) Entropy ($-\Sigma{\frac{f_i}{L_1}\text{log}\frac{f_i}{L_1}}$), and Cardinality ($|\langle f_1,\dots, f_N \rangle|$). Note that there are many  other \statistics that can be estimated by combining statistics, such as standard deviation, histograms, mean, or median. A statistic that cannot be directly estimated by \sketch is quantiles.}

The basic building block of universal sketches are L2-Heavy Hitter (HH)        sketches
\eg Count-sketch~\cite{cs1}. Each count-sketch maintains   $\sketchrows_{\CountSketch}$
arrays of $\sketchwidth_{\CountSketch}$ counters each,     $\sketchrows_{\CountSketch}$
pairwise-independent hash functions and a max-heap keeping track of the top-$k$   Heavy
Hitters in the sketch; When updating each count-sketch with a new  data item, the sketch
updates a randomly located counter in every row based on the corresponding hash index to
keep track of that data item's frequency. The top-$k$ HH heap is subsequently    updated
to reflect the addition of the new item. A universal
sketch consists of $\NumLayers$ layers of count-sketches. Each count sketch applies an
independent 0-1 hash function $\hash_{\layerindex\in[0,\ \NumLayers)}$  to the   input
data stream to sub-sample at every layer (from the previous layer).
These layers then track the heavy hitters, \ie the important contributors    to
the $\Gstatistic$-sum. 

The intuition here is that the layered structure of  the universal sketch is designed     for
sampling representative elements with diverse frequencies, and these elements can be  used
to estimate \Gstatistic-sum with bounded errors. If only one layer of heavy hitter sketch
were  used, the estimations would lack representatives from less frequent elements.
The heavy-hitters at each layer are processed iteratively from the bottom layer to   the
top and the recursively aggregated result is used to compute the desired statistic. This
is an unbiased estimator of \Gstatistic-sum with bounded additive errors (Theorem~\ref{theorem:1}).

\begin{theorem}[\cite{univmon, zero_one_freq_law}]
Given a stream  $\Datastream_{\streamlength, \streamdistinct}$ let us  consider a Universal Sketch $\UnivSketch$  with $\NumLayers= \mathcal{O}(\log \streamdistinct)$ layers. If each layer of $\UnivSketch$ provides an   $(\epsilon_{\UnivSketch},\ \delta_{\UnivSketch})$-L2 error guarantee, then $\UnivSketch$     can
estimate any $\Gstatistic$-sum function $\Gstatistic   \in \text{Stream-Polylog}$
to within a   $(1\pm\epsilon_{\UnivSketch})$  factor with probability
$1-\delta_{\UnivSketch}$. Satisfying a $(\epsilon_\UnivSketch,\      
\delta_\UnivSketch)$-L2 error guarantee requires $O(log \streamdistinct)$    Count-Sketch
instances with $w_{\CountSketch}=\mathcal{O}(\epsilon_\UnivSketch^{-2})$ columns
and $r_{\CountSketch}=\mathcal{O}(\log \delta_\UnivSketch^{-1})$ rows.
\label{theorem:1}
\end{theorem}

\subsection{The \sketch Algorithm} 
\label{subsec:operations}

Combining these ideas gives us the \sketch algorithm.

\begin{enumerate}[align=left, leftmargin=0pt, labelindent=\parindent, listparindent=\parindent, labelwidth=0pt, itemindent=!]
\item \textbf{Updating \sketch:} 
Updating \sketch with a \datapoint,
$\datapointvar_\itemindex = <\dimensionval_{1, \itemindex},\ \dimensionval_{2,\itemindex},\
\dots, \dimensionval_{\NumDimensionsPerData, \itemindex,},\  \metricval_\itemindex>$ is        a
three-step process.  At the first, ``fan-out'' stage, we compute          the
$O(2^\NumDimensionsPerData)$ \aggregations $\{\aggregationval_1,        \dots,
\aggregationval_{2^\NumDimensionsPerData}\}$ that $\datapointvar_\itemindex$
belongs in. \antonis{Note that while $O(2^\NumDimensionsPerData)$ is an exponential term, it is exponential to the number of dimensions $\NumDimensionsPerData$ and, thus, significantly smaller than the total number of \aggregations $\aggregationval$, which is exponential to the cardinality of values in each dimension.} Then, we map each $\aggregationval_\aggregationindex$         to
$\sketchrows$ universal sketches  instances using $\sketchrows$ pairwise-independent
hash functions  $\hash_{\rowindex\in[0, \sketchrows)}: \aggregationval_\aggregationindex
\rightarrow    [0,\sketchwidth)$.   
For the $\rowindex^{th}$ row, the index of the universal sketch to update $\UnivSketch_\rowindex$
is the hash of $\aggregationval_\aggregationindex$ using hash function $\hash_\rowindex$.
Last, we update each $\UnivSketch_\rowindex$ with the metric value $\metricval_j$.

\item \textbf{Querying \sketch:} \sketch's querying algorithm  takes as input a  statistic
$\gstatistic$ and an aggregation $\aggregationval_\aggregationindex$ \ie the aggregation to
estimate $\gstatistic$ on. Querying consists of 2 steps. The first involves identifying the
set  of $\sketchrows$ universal sketch instances $\{\UnivSketch_\rowindex\}$ that 
$\aggregationval_\aggregationindex$ maps to. Then  $\gstatistic$ is estimated from each 
$\UnivSketch_\rowindex$, and the median value of these estimations is returned. 

\end{enumerate}

\antonis{Given this basic algorithm, we now focus on formally proving that \sketch offers rigorous accuracy guarantees and that it is usable in practice.}

\subsection{Accuracy Guarantees}
\label{subsec:proof}

 Theorem~\ref{theorem:2} states the accuracy  bounds
of \sketch. 
\begin{theorem}
Let us assume that each Universal Sketch $\UnivSketch$ can approximate the $\Gstatistic$-sum,
for a monotone function $g$ within a $(1+\epsilon_{\UnivSketch})$-factor with    probability
$1-\delta_{\UnivSketch} > 1/2$. Further, let $\Gstatistic_{\Datastream}$ be the $\Gstatistic$-sum  
applied to the  stream $\Datastream$ and $\Gstatistic_i$ when applied to the target \aggregation
$\mathcal \aggregationval_i$. Then \sketch with $w=\mathcal{O}(\epsilon^{-1})$ columns and   $r=\mathcal{O}(\log
\delta^{-1})$ rows, for user defined parameters $\epsilon,\ \delta$, provides an    estimate
$\widehat{\Gstatistic_i}$ that with probability $1-\delta$ satisfies:
\begin{align}
\Gstatistic_i(1-\epsilon_{\UnivSketch})\le\widehat{\Gstatistic_i}\le \Gstatistic_i(1+\epsilon_{\UnivSketch})+\epsilon\cdot \Gstatistic_{\Datastream}
\label{eq:1}
\end{align}
\label{theorem:2}
\end{theorem}
\begin{proof}
To bound the error of our algorithm, we analyze the frequency vector $f_j$ of the stream of elements    mapped to each Universal Sketch instance $\UnivSketch_j=\hash_j(\aggregationval_i)$, where $\aggregationval_i$ is the queried \aggregation. The frequencies of all $\metricval_i \in \aggregationval_i$ are guaranteed to   appear in $f_j$, since the {\sc Update} algorithm of~\S\ref{subsec:operations} \mbox{maps them  to
$\UnivSketch_j$. }

Let $\mathfrak Q = \set{\aggregationval_1,\ldots,}$ denote all groups in the input stream
$\Datastream$, and let $\mathfrak Q_j=\set{\aggregationval_k\in \mathfrak Q\mid \hash_j(\aggregationval_k)=
\hash_j(\aggregationval_i)}$ denote the set of groups mapped to $\UnivSketch_j$. That is, 
$\Gstatistic_{\Datastream}=\sum_{\aggregationval_k\in\mathfrak Q}\sum_{\metricval_k\in\aggregationval_k}\gstatistic(f_{\metricval_k})$.

The quantity which we wish to estimate is $\Gstatistic_i\triangleq \sum_{x\in\aggregationval_i}\gstatistic(f_\metricval)$, \ie
the $\gstatistic$-sum of the group $\aggregationval_i$, while the $\UnivSketch_j$ processes all groups in  $\mathfrak Q_j$ and
thus approximates $\sum_{\aggregationval_k\in\mathfrak Q_j}\sum_{\metricval_k\in\aggregationval_k}g(f_{\metricval_k})= \Gstatistic_i + \sum_{\aggregationval_k\in\mathfrak Q_j\setminus\set{\aggregationval_i}}\sum_{\metricval_k\in\aggregationval_k}g(f_{\metricval_k})$.
For all $j\in\set{0,\ldots,r-1}$, denote by $\widehat{\Gstatistic_{i,j}}$ the estimate of $\UnivSketch_j$, and denote the \emph{noise} added by the other groups as $N_j=\sum_{\aggregationval_k\in\mathfrak Q_j\setminus\set{\aggregationval_i}}\sum_{\metricval_k\in\aggregationval_k}g(f_{\metricval_k})$.
Notice that, since any group in $\mathfrak Q\setminus\set{\aggregationval_i}$ has a probability of $1/\sketchwidth$ of being in $\mathfrak Q_j$, its expectation satisfies that: 
$
\mathbb E[N_j]= \frac{\sum_{\aggregationval_k\in\mathfrak Q\setminus\set{\aggregationval_i}}\sum_{\metricval_k\in\aggregationval_k}g(f_{\metricval_k})}{\sketchwidth}\le \frac{\Gstatistic_{\Datastream}}{\sketchwidth}.
$
Therefore, according to Markov's inequality, for any $c\in\mathbb R^+$, $\Pr[N_j\ge c\cdot \frac{\Gstatistic_{\Datastream}}{\sketchwidth}]\le 1/c$.
Next, by the correctness of the universal sketch, we have that, $$\Pr[\widehat{\Gstatistic_{i,j}}\notin[(\Gstatistic_i+N_j)(1-\epsilon_{\UnivSketch}),(\Gstatistic_i+N_j)(1+\epsilon_{\UnivSketch})]]\le\delta_{\UnivSketch}.$$
Since $\gstatistic$ is part of $G$-sum $\in$ {\em Stream-PolyLog}, it must be monotone, and thus $N_j\ge 0$. This means that with probability of at least $1-\delta_{\UnivSketch}-1/c$ both $\widehat{\Gstatistic_{i,j}}\in[\Gstatistic_i(1-\epsilon_{\UnivSketch}),(\Gstatistic_i+N_j)(1+\epsilon_{\UnivSketch})]$ and $N_j< c\cdot \frac{\Gstatistic_{\Datastream}}{\sketchwidth}$ simultaneously hold, and thus
\begin{multline}
    \Gstatistic_i(1-\epsilon_{\UnivSketch})\le\widehat{\Gstatistic_{i,j}}\le \Gstatistic_i(1+\epsilon_{\UnivSketch})+\frac{c}{\sketchwidth}(1+\epsilon_{\UnivSketch})\Gstatistic_{\Datastream}.
    \label{eq:w}
\end{multline}

Therefore, we pick $\sketchwidth=c\cdot (1+\epsilon_{\UnivSketch})\cdot \epsilon^{-1}$ and a $c$ value such that $1-\delta_{\UnivSketch}-1/c > 1/2$, to get that 
$$
\Pr\brackets{\Gstatistic_i(1-\epsilon_{\UnivSketch})\le\widehat{\Gstatistic_{i,j}}\le \Gstatistic_i(1+\epsilon_{\UnivSketch})+\epsilon\cdot \Gstatistic_{\Datastream}}> 1/2
$$
Recall that the algorithm's query sets $\widehat{\Gstatistic_i}=\text{median}_j \widehat {\Gstatistic_{i,j}}$ and that
the $\sketchrows$ rows are i.i.d. and thus a  Chernoff bound yields  
$$
\Pr\brackets{\Gstatistic_i(1-\epsilon_{\UnivSketch})\le\widehat{\Gstatistic_i}\le \Gstatistic_i(1+\epsilon_{\UnivSketch})+\epsilon\cdot \Gstatistic_{\Datastream}}\ge 1-\delta.\qedhere
$$
\end{proof}


\mypara{Takeaways} We note the following from Theorem~\ref{theorem:2}. The error bounds of \sketch are tunable based on the choice of its configuration parameters that control $(\epsilon,\ \delta)$ and $(\epsilon_{\UnivSketch},\ \delta_{\UnivSketch})$. In addition, the upper error bound is additive, which means that it will allow for loose error bounds in cases where $\epsilon\cdot\Gstatistic_{\Datastream} \approx \Gstatistic_i$. We discuss these takeaways in more detail below.

\subsection{\sketch Configuration}
\label{subsec:configs}

\begin{figure}[thb]
        \includegraphics[width=1\linewidth, height=70pt]{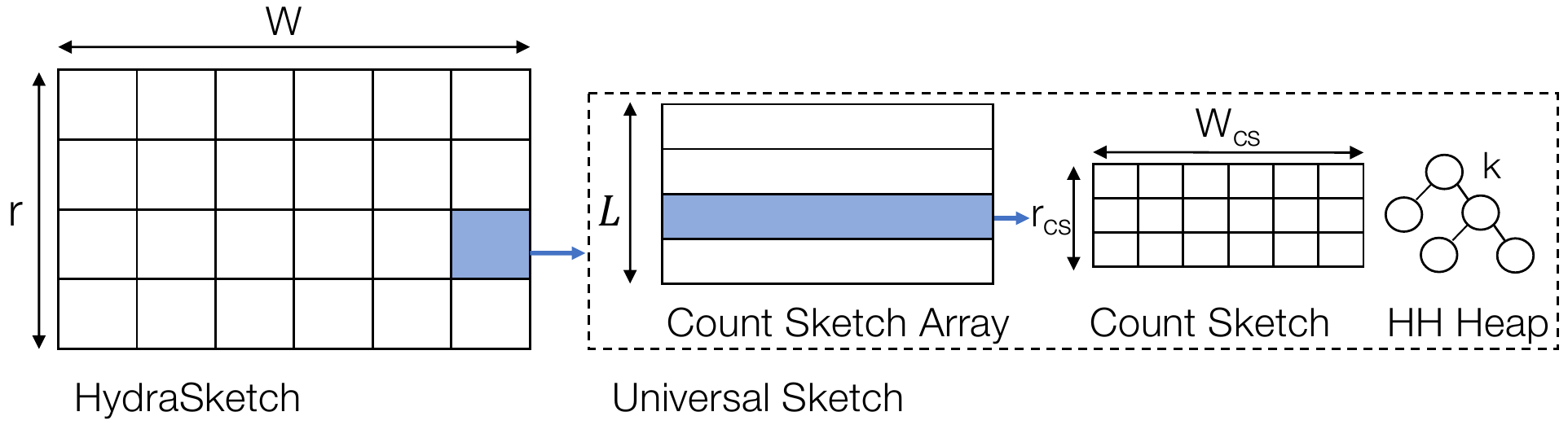}        \caption{\sketch structure and configuration parameters.}
        \label{fig:sec4:params}
\end{figure}

We now focus on techniques to tune \sketch's parameters. As illustrated in Figure~\ref{fig:sec4:params}, \sketch has six configuration parameters: two parameters ($\sketchrows$ and $\sketchwidth$) define the structure of the sketch arrays and additional four ($\NumLayers,\ \sketchwidth_{\CountSketch},\ \sketchrows_{\CountSketch},$ and $\heavykeys$) determine the inner structure of the Universal Sketches. The choice of configuration parameters of \sketch affects its empirical accuracy and memory footprint. For instance, larger $\sketchwidth$ and $\sketchrows$ values ensure better estimation accuracy but require more memory.

It is often useful to reason about the \emph{relative error} of the estimation; rephrasing Theorem~\ref{theorem:2}, we can write:
{\small $$
\Pr\bigg[{-\epsilon_{\UnivSketch}\le\frac{\widehat{\Gstatistic_i}-\Gstatistic_i}{\Gstatistic_i}\le \epsilon_{\UnivSketch}+\epsilon\cdot \frac{\Gstatistic_{\Datastream}}{\Gstatistic_i}}\bigg]\ge 1-\delta.
$$
}
and thus
{\small $$
\Pr\brackets{\frac{\left|\widehat{\Gstatistic_i}-\Gstatistic_i\right|}{\Gstatistic_i}\le \epsilon_{\UnivSketch}+\epsilon\cdot \frac{\Gstatistic_{\Datastream}}{\Gstatistic_i}}\ge 1-\delta.
$$}
That is, we have that with probability $1-\delta$, the relative error is at most $\epsilon_{\UnivSketch}+\epsilon\cdot \frac{\Gstatistic_{\Datastream}}{\Gstatistic_i}$. Since $\epsilon_{\UnivSketch},\ \epsilon$, and $\Gstatistic_{\Datastream}$ are determined by the configuration and not a specific subpopulation, we get that the relative error bound is looser if $G_i$ is small. Intuitively, if a subpopulation is very small, the noise we get from the colliding subpopulations may be larger than its own statistics.

With that in mind, we consider a quantity $G_{\min}$ that denotes the minimal G-sum for which we want to guarantee some relative error, \eg of 20\%, with a high probability, \eg 90\%. This means that for any subpopulation with a higher G-sum, the error is upper bounded by $\epsilon_{\UnivSketch}+\epsilon\cdot \frac{\Gstatistic_{\Datastream}}{G_{\min}}$. This allows us
to derive configuration heuristics for \sketch as follows:

\mypara{Controlling the probability of error bounds holding} From Theorem~\ref{theorem:2}, for the error bound of our example to hold with 90\% probability, $1-\delta=0.9$ and, hence,  $\delta= 0.1$. This translates to $\sketchrows  \approx 3$. Similarly, from Theorem~\ref{theorem:1}, a universal sketch will estimate any $\Gstatistic$-sum function within an $\epsilon_\UnivSketch$ factor with probability $1-\delta_\UnivSketch$. For probability $90\%$, $\delta_\UnivSketch = 0.1$ and, thus, $\sketchrows_\CountSketch \approx3$.

\mypara{Minimizing upper error bound} To minimize the upper error bound of \sketch, we need to minimize $E = \epsilon_\UnivSketch + \epsilon \frac{\Gstatistic_\Datastream}{\Gstatistic_{min}}$ under a memory constraint, $\mathcal{O}(M) = \sketchwidth\times\sketchwidth_\UnivSketch$. From Theorems~\ref{theorem:1} and~\ref{theorem:2}, we know that $\epsilon\approx 1/\sketchwidth$ and $\epsilon_\UnivSketch\approx1/\sqrt{\sketchwidth_\UnivSketch}$. This allows us to minimize $E$  for $\sketchwidth$ and $\sketchwidth_\UnivSketch$ as follows:

\begin{packedenumerate}
\item \mypara{Solving for $\epsilon_\UnivSketch$} Given the memory constraint, we can write $E=\epsilon_\UnivSketch + \frac{\Gstatistic_\Datastream}{M\Gstatistic_{min}\epsilon_\UnivSketch^2}$.
Minimizing $E$ over $\epsilon_\UnivSketch$ gives us: 
\begin{align}
\epsilon_\UnivSketch = \sqrt[3]{\frac{2\Gstatistic_\Datastream}{M\Gstatistic_{min}}} \Rightarrow  \sketchwidth_\UnivSketch = \Theta(\frac{M\Gstatistic_{min}}{\Gstatistic_\Datastream})^{2/3}.
\label{eq:h1}
\end{align}

\item \mypara{Solving for $\epsilon$} Similarly, we can write $E= \sqrt{\frac{1}{M\epsilon}} + \epsilon\frac{\Gstatistic_\Datastream}{\Gstatistic_{min}}$. Minimizing over $\epsilon$ gives:
\begin{align}
\epsilon = \parentheses{\frac{2\sqrt{M}\Gstatistic_{\Datastream}}{\Gstatistic_{min}}}^{-2/3} \Rightarrow \sketchwidth = \Theta\parentheses{\frac{\sqrt{M}\Gstatistic_\Datastream}{\Gstatistic_{min}}}^{2/3}
\end{align}
\end{packedenumerate}

\mypara{Controlling remaining universal sketch parameters} 
Last, we configure the levels ($\NumLayers$) maintained in each universal sketch instance and the number of heavy keys ($\heavykeys$) needed to store at each level's heavy hitter heap. From Theorem~\ref{theorem:1}, $\NumLayers=\mathcal{O}(\log{\streamdistinct_{\UnivSketch}})$, where $\streamdistinct_{\UnivSketch}$ is the average number of  distinct \aggregations summarized at a universal sketch. For the value of $\heavykeys$, we empirically set its lower bound to $\heavykeys = \Omega(1/\epsilon_\UnivSketch^2)$. For $\epsilon_\UnivSketch=0.1$, this translates to $k \approx 100$.

Let us now see how we can use these guidelines in practice. As an example, let us assume we want the relative error of estimation to not exceed $0.2$ with 90\% probability for \aggregations where $\Gstatistic_i/\Gstatistic_\Datastream \geq 10^{-3}$. Thus, $\Gstatistic_{min}=10^{-3}\cdot \Gstatistic_\Datastream$. Let us also assume that $\epsilon_\UnivSketch = \epsilon\cdot\frac{\Gstatistic_\Datastream}{\Gstatistic_{min}}=0.1$. From Eq.~\eqref{eq:h1}, we can get an estimate of memory needed, $M\approx10^6$ needed. Note that here $M$ measures ``units of $\sketchwidth_\UnivSketch$'' \ie counters. Thus, $\epsilon_\UnivSketch=0.1$ and  $\sketchwidth_\UnivSketch=\Theta(10^2)$. From $\mathcal{O}(M) = \sketchwidth\cdot\sketchwidth_\UnivSketch$, we can further see that also $\sketchwidth=\Theta(10^2)$. 

\antonis{In \S\ref{sec:eval}, we show that these strategies can  achieve near optimal tradeoffs.    We acknowledge that implementing this workflow  assumes  that the operator has some    prior knowledge about the workload \ie a rough estimate of the number of subpopulations.  We believe this is not an unreasonable requirement in many  practical settings.}

\section{Implementation} 
\label{sec:query}

This section discusses our implementation of \sysname and  the practical performance  challenges we faced. Our prototype of \sketch can be found in~\cite{repo}.

\mypara{Baseline Implementation and Workflow} We implement \sysname's workflow (\S\ref{sec:overview}) on top of Apache Spark~\cite{spark} as Spark's extensibility allowed us to easily prototype design alternatives. However, \sysname's workflow can easily fit into different analytics frameworks \eg Druid~\cite{druid}.

Data ingestion happens at the worker nodes. Each worker node splits its input into $\sim$64MB partitions, allocates one \sketch instance per partition and updates it with that partition's  data. We implement these and \sketch instances as Spark RDDs. To allocate appropriately configured \sketch instances, workers rely on configuration manifests distributed by the \primary node.

As a result of splitting input data into smaller batches, each worker node maintains multiple instances  of \sketch. \antonis{The design of \sysname enables sketch merging due to the well-known linearity property of frequency-based sketches. Therefore, during data ingestion, worker nodes merge \sketch instances  of fully ingested partitions until \sysname
is left with one \sketch instance to query.} For sketch merging, we use Spark's ``treeAggregation'' module~\cite{treeagg}, thus mitigating the risk of  performance bottlenecks.

Query estimation involves both the \primary and the worker nodes. The operator inputs the desired queries and the \primary then generates a query plan for the worker nodes to execute. Estimation results are collected at the \primary node.

\mypara{An accuracy-improving heuristic}
Recall from \S\ref{subsec:operations} that after $\aggregationval_i$ is mapped to a universal sketch, that sketch only stores the frequencies of metric values $\metricval_j$. This design, however, does not keep track of which \aggregation $\aggregationval_i$ each $\metricval_j$ maps to. As a result, a universal sketch will return the same estimations for all \aggregations whose data it stores. Our heuristic is simple:  Instead of updating each universal sketch with $\metricval_j$, we can use a more fine-grained key, \ie the concatenation of the metric value and its corresponding \aggregation. 
This way, heavy hitter heaps will maintain heavy counts for each ($\aggregationval_i,\ \metricval_j$) pair and will    be
able to differentiate between them.

\mypara{Implementation optimizations}
To further reduce the system's runtime, we introduce a few optimizations:

\begin{enumerate}[align=left, leftmargin=0pt, labelindent=\parindent, listparindent=\parindent, labelwidth=0pt, itemindent=!]
\item \textbf{One Large Hash per ($\aggregationval_i$, $\metricval_j$) Pair:} 
Updating \sketch ($\aggregationval_i$, $\metricval_j$)  requires
$\mathcal{O}(\sketchrows\times\NumLayers)$ hash computations, $\sketchrows$ to identify
the universal sketches to update and up to $\NumLayers$ per universal  sketch.
We reduce the number of hashes to $\mathcal{O}(1)$   by
computing one large 128-bit hash and breaking it down into substrings      of
variable lengths and treating each substring as a separate hash. Prior analysis~\cite{flajolet1985probabilistic,kirsch2006less} shows that different substrings from the same long hash provide sufficient independence.

\item \textbf{One Layer Update:} In prior universal sketching implementations, 
the algorithm keeps a heap to track frequent keys per layer. For each datapoint update, the universal sketching needs to update two of its layers on average.
In \sketch, we follow~\cite{joltik} and update only the lowest sampled layer per datapoint. This technique reduces the layers updated to one per datapoint, while providing an equivalent implementation.

\item \textbf{Heap-only sketch merge} Merging two \sketch instances involves iterating over two 2D universal sketches arrays  $\text{HS}^1$ and $\text{HS}^2$ and merging each pair $(\text{HS}^1_{k,l},\ \text{HS}^2_{k,l})$. This means iterating over the universal sketch layers, summing up corresponding counters, recomputing the heavy elements and re-populating the heavy hitter heaps. However, we find that we can only merge the heavy hitter heaps \mbox{instead of all counters.} 

\end{enumerate}

\section{Evaluation}
\label{sec:eval}

We now evaluate \sysname  using real-world and synthetic datasets. We provide a sensitivity analysis of our design, and evaluate our configuration strategies and optimizations. In summary:
\begin{packedenumerate}


\item \sysname offers $\leq 10$sec query latencies and is $7$-$20\times$ smaller than existing analytics engines. 

\item \sysname offers $\leq $5\% mean errors (combined across statistics) with 90\% probability for a broad set of \statistics at  1/10 of the \$ cost of exact analytics engines.

\item Thanks to \sysname's sub-linear (to the number of \aggregations) memory scaling, \sysname achieves close to an order of magnitude improvement in operational cost compared to the best exact analytics baseline. 

\item \sysname's  sketch configuration strategies ensure near-optimal \mbox{memory-accuracy tradeoffs.}

\item \sysname's performance optimizations improve   end-to-end system runtime by 45\% compared to a deployment that uses  the basic \sketch design.  
\end{packedenumerate}


\subsection{Experimental Methodology}

 \mypara{Setup} We evaluate \sysname on a 20-node cluster of m5.xlarge (4CPU - 16GB memory) AWS servers~\cite{pricing}. In practice, we observe that nodes have $\approx$10-11GB of available main memory. We allocate 3 CPUs for \sysname and its input data is CSV files that are streamed from AWS S3.
We configure \sketch  using the heuristics of \S\ref{subsec:configs} to ensure a conservative lower error bound of -10\% (\ie $\epsilon_\UnivSketch=0.1$) and upper bound of 20\% with 90\% probability for $\Gstatistic_{min}/\Gstatistic_\Datastream=2\cdot10^{-3}$. We also use the performance optimizations of \S\ref{sec:query}. While these bounds are  conservative, they ensure a memory footprint of $< 100$MB per \sketch instance;  our results show that the actual errors were much smaller.

 \mypara{Datasets}  We use two real-world datasets and a synthetic trace. Each dataset maps to a different usecase. First, we use CAIDA flow traces~\cite{caida} collected at a  backbone link of a Tier1 US-based ISP. The total trace is up to 130GB in initial size and flow  data can be clustered in up to approximately 5.6M subpopulations $\aggregationval_i$. Given that we analyze $\metricval_j$ metric values per subpopulation, this dataset contains up to 506M distinct $\langle\aggregationval_i, \metricval_j\rangle$ pairs. Second, we use a real-world  trace of video session summaries  corresponding to one major US-based streaming-video provider. The size of the video-QoE trace is approximate 5GB, with data that we cluster in up to 700k \aggregations and up to 25M $\langle\aggregationval_i, \metricval_j\rangle$ pairs. Third, we generate synthetic traces following Zipf distribution with varying skewness (\eg 0.7 to 0.99).

 \mypara{Summary statistics} We evaluate \sysname's accuracy using L1/L2 norms, entropy and cardinality \ie statistics that map to the queries described in
\S\ref{sec:motivation}. For each \aggregation, we compute the precise value of each statistic as  ground truth and then estimate the relative error with respect to \sysname's accuracy.

%
%

\begin{figure*}[t]
\centering
\subfigure[L2-Norm.]{
\includegraphics[width=0.24\textwidth]{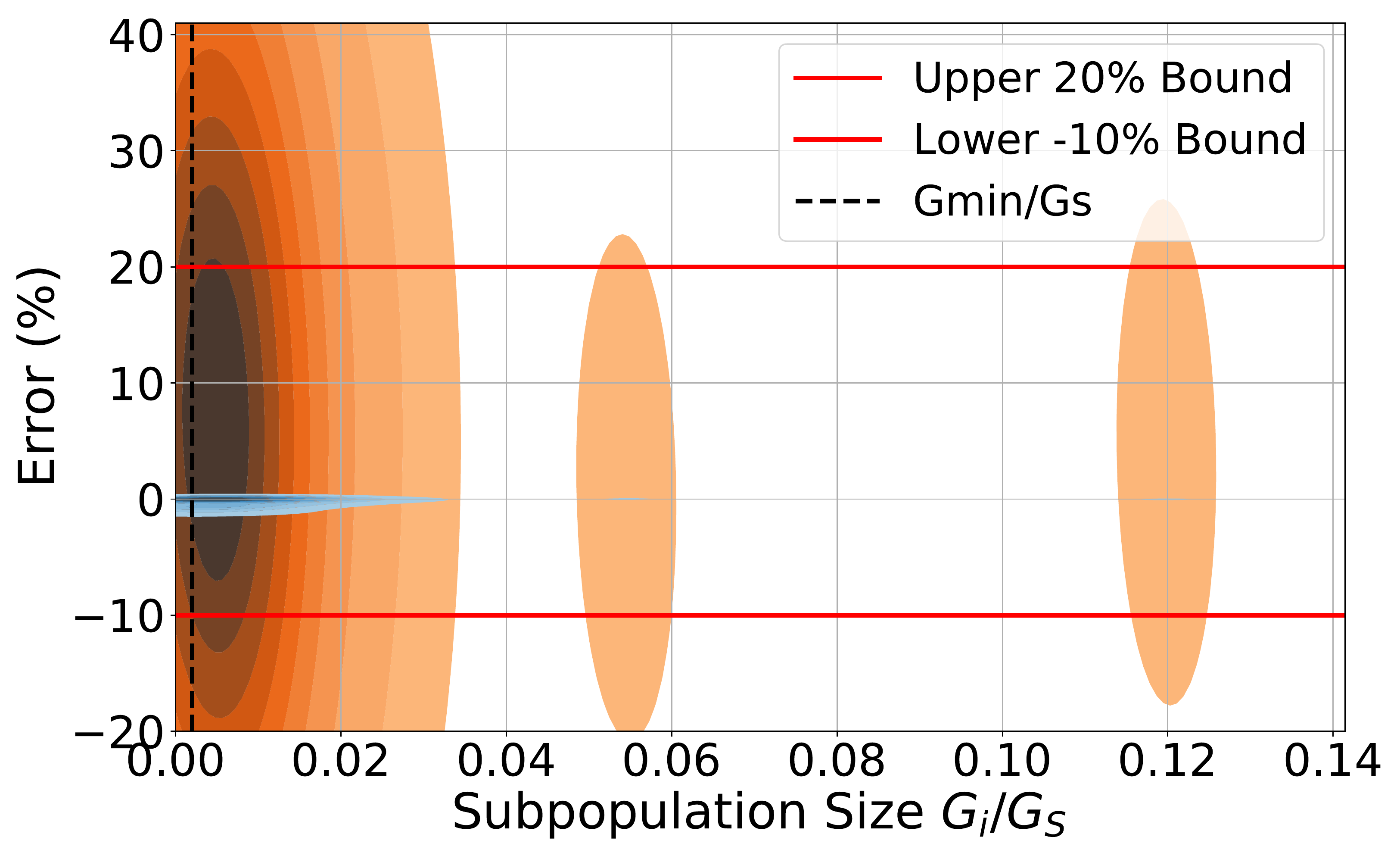}\label{fig:kde_l2}}
\subfigure[Entropy.]{
\includegraphics[width=0.24\textwidth]{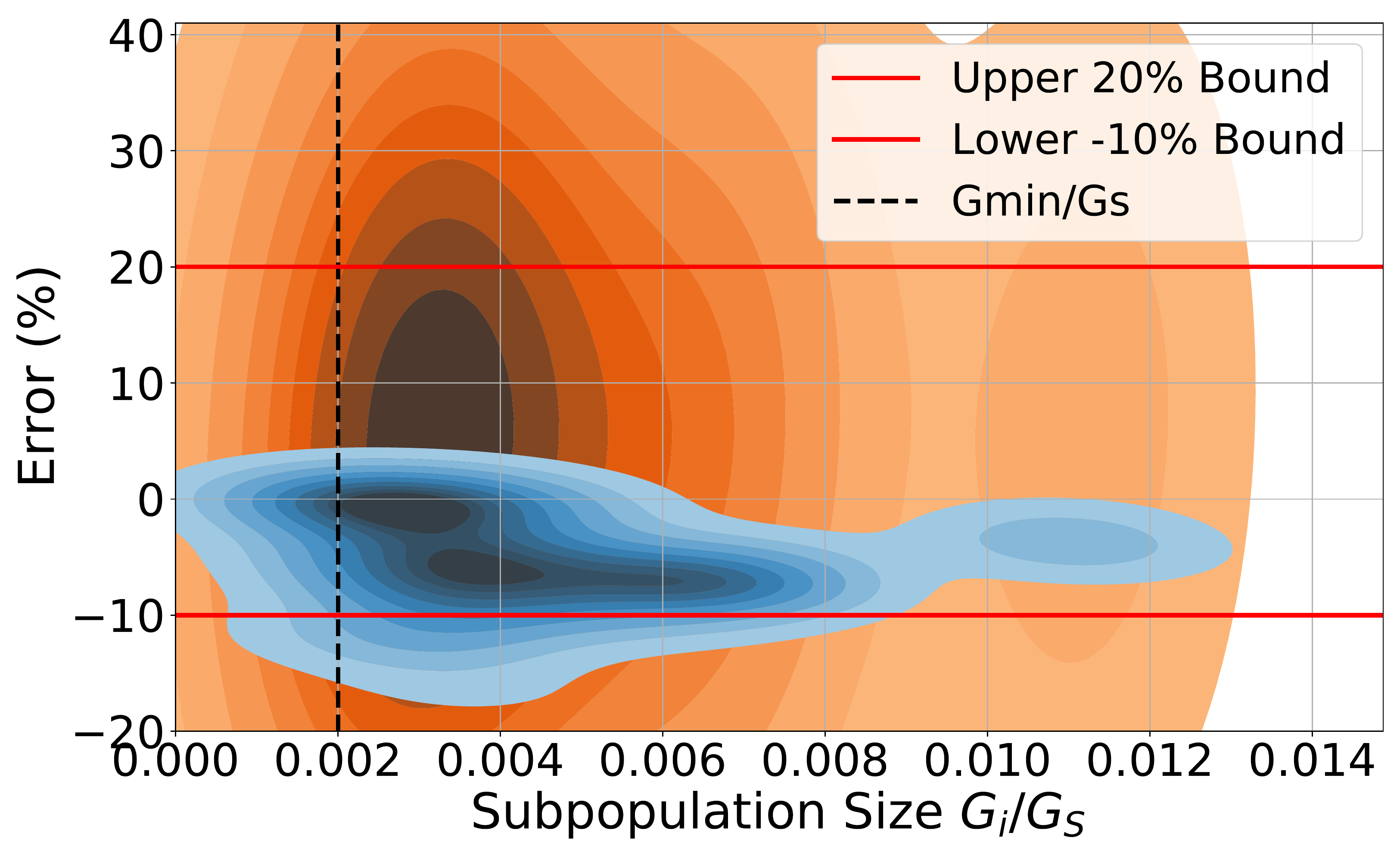}
\label{fig:nf-chain}}
\subfigure[L1-Norm.]{
\includegraphics[width=0.24\textwidth]{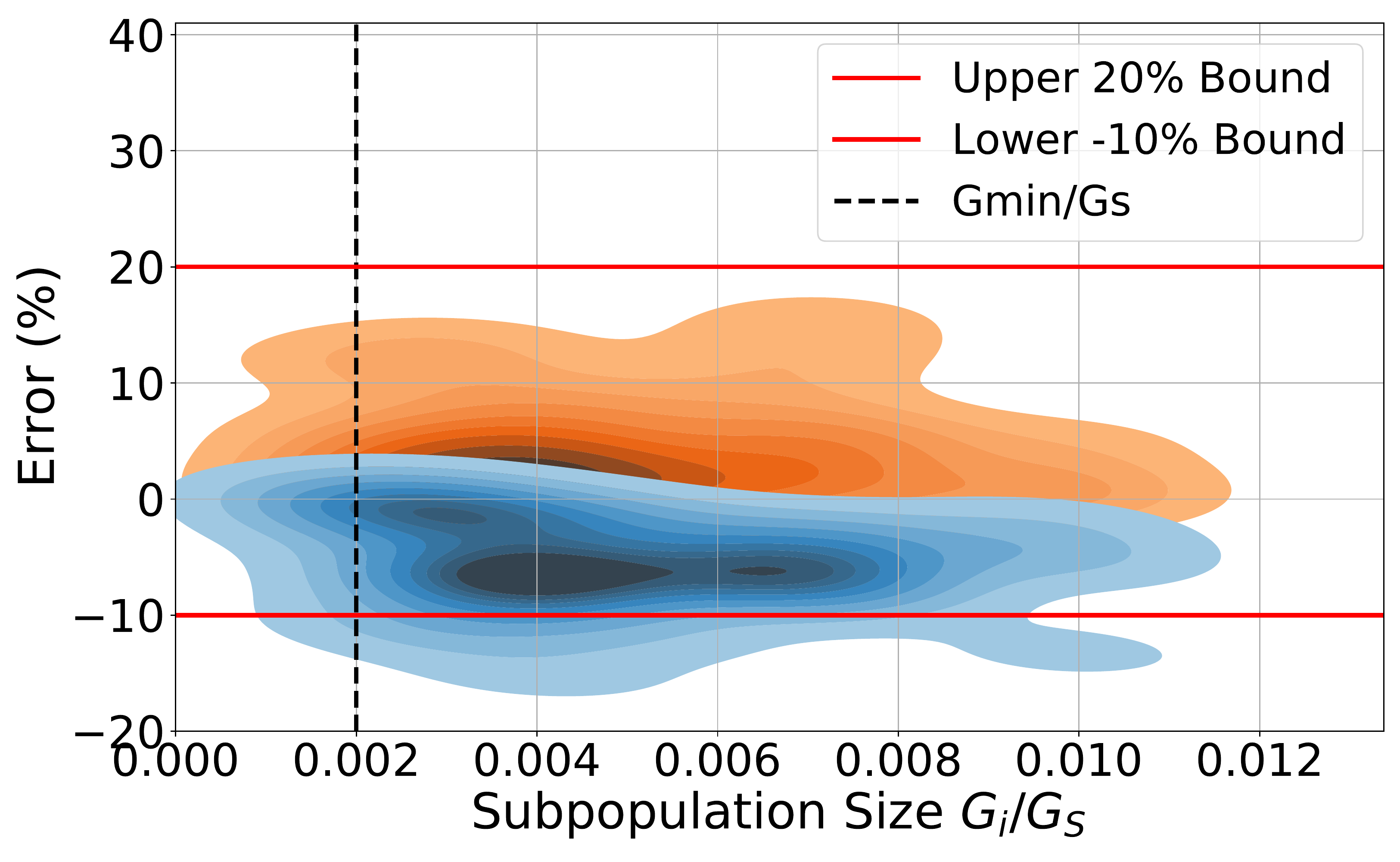}\label{fig:latency}}
\subfigure[Cardinality.]{
\includegraphics[width=0.24\textwidth]{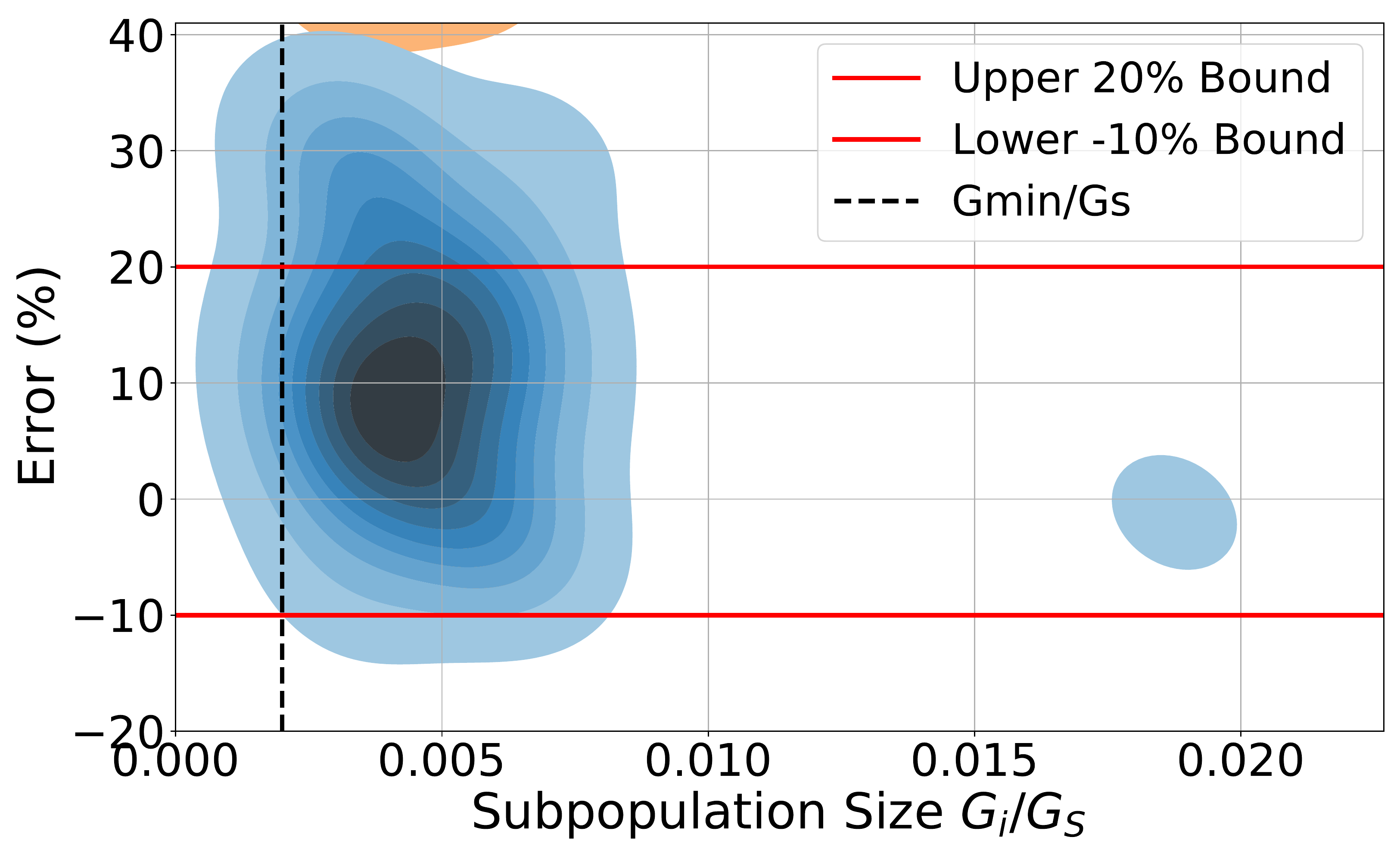}\label{fig:latency}}
\caption{\antonis{Error distribution for different data subpopulations for \sysname (blue) and uniform sampling (orange).
Red lines indicate the error threshold.}}
\label{fig:kde}
\end{figure*}


\mypara{Evaluation baselines} For our experiments, we compare \sysname against several baselines:
From the space of precise analytics we compare with: (1) \textbf{Spark-SQL:} This is a traditional
SQL implementation where incoming \datapoint is stored as a row in one   (logical)
data table. At estimation time, we create a Key-Value store, where the keys are
distinct subpopulations $\aggregationval_\aggregationindex$ and the values are lists of metric
values $\metricval_j$ per subpopulation;  (2) \textbf{Spark-KV:} Here, we summarize incoming
data at ingestion time and maintain a Key-Value store where the keys       are
distinct $\langle\aggregationval_i,\ \metricval_j\rangle$ pairs and   the
values are their respective frequency counts; (3) \textbf{Druid:} This is similar to Spark-KV but uses Druid's data roll-up feature
to generate the key-value store. 

From the space of approximate analytics engines, we compare against:
(1) \textbf{Uniform Sampling:} We implement 10\% uniform sampling at ingestion time and then
apply the Spark-KV approach to the sub-sampled data that contains $\approx82$M distinct $\langle\aggregationval_i, \metricval_j\rangle$ pairs;
(2) \textbf{VerdictDB~\cite{verdictdb}:} We deploy VerdictDB on Amazon Redshift and use the default nodes of that service (20 dc2.large nodes, each with 2CPU, 15GB memory and 160GB NVMe-SSD as storage) as backend SQL engine. VerdictDB builds offline samples, so we create hash-based sample tables for cardinality metric and uniform sample tables for L1 and L2 norm. We set sampling rate = 1\% for both sample tables. VerdictDB does allow  entropy estimations; (3) \textbf{One Universal  Sketch per \aggregation}. 


\subsection{End-to-End Evaluation of \sysname}

To evaluate \sysname end-to-end we investigate whether the system meets  operators' requirements as outlined in \S\ref{sec:motivation}. To that end, we investigate three questions:

\myparaq{What is \sysname's operating cost compared to our baselines}
We measure the normalized query estimation \$ cost for 4 statistics for the CAIDA dataset (130GB, 5.6M subpopulations, 506 distinct $\langle\aggregationval_i, \metricval_j\rangle$ pairs). We estimate their normalized cost as VerdictDB on Amazon Redshift constrained us to specific servers with a different pricing model. 

Figure~\ref{fig:sec7:cost} depicts \sysname's cost-accuracy tradeoff. \sysname's cost is $\sim$2 orders of magnitude smaller than that of Spark-SQL. That is because Spark-SQL processes the entire dataset at query time and because estimation happens at the \primary node. \sysname's  estimation cost is also an order of magnitude lower than  Druid's which uses data summaries created at ingestion. However, as we will see later, Druid's ingestion is very inefficient.  The best performing, precisely accurate baseline is Spark-KV that produces frequency counts for the resulting 506 KV-pairs at ingestion time and uses that for estimating statistics. Spark-KV is $\sim$7$\times$ more expensive than \sysname. 

Regarding approximate analytics baselines, we observe that VerdictDB, while very accurate ($\sim$98\% mean accuracy for 1\% sampling, exhibits large estimation times, comparable to worst-case estimation times in the original VerdictDB paper~\cite{verdictdb}. When normalized by server cost, VerdictBD's cost is comparable to Spark-SQL. \sysname's operational cost is on par with   a sampling approach that uniformly samples 10\% of all data but whose error can be
very large. Perhaps surprisingly, the 10\% baseline exhibits higher cost. This is because this baseline still needs to process $\approx82$M KV pairs and still requires more memory than \sysname. In the case of the smaller \conviva  dataset (not shown due to lack of space), \sysname is only 3$\times$ cheaper than Spark-SQL and approximately as costly as Spark-KV. This smaller gap is due to the smaller size of the dataset. 
In \S\ref{subsec:analysis}, we look at the empirical runtime and memory requirements that explain the observed cost results.

\begin{figure}[tb]
        \includegraphics[width=1\linewidth]{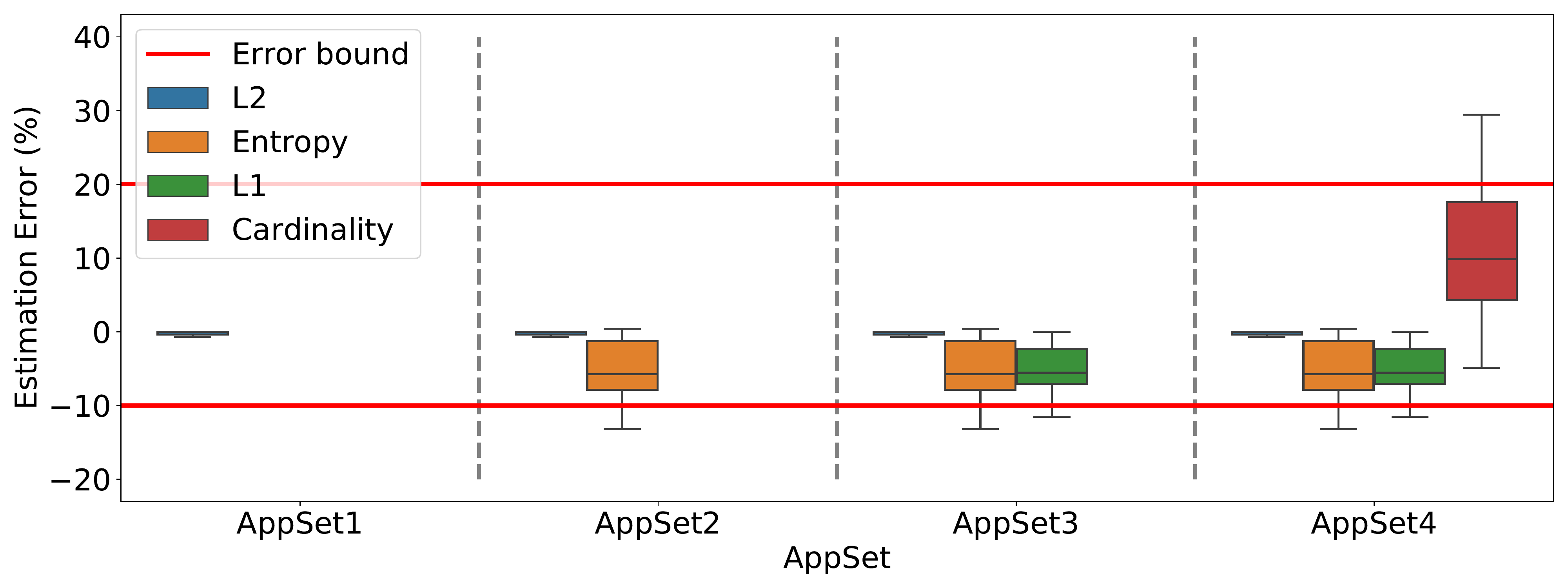}
        \caption{\sysname's estimation error for the CAIDA dataset.}
        \label{fig:sec7:caidaAppSet}
\end{figure}

\myparaq{Does \sysname enable interactive query latencies}
Figure~\ref{fig:sec7:caidaruntime} illustrates \sysname's runtime  as a function   of
the dataset size and the number of data  \aggregations for the CAIDA dataset.
We can see that \sysname's query time is $\sim$11sec for 5.6 million data \aggregations,
almost one order of magnitude ($7\times$) smaller than that of Spark-KV.  \antonis{We find this to be an acceptable query latency for a framework that is configured to periodically run estimations on streaming data (\eg every minute) and large volumes of \aggregations.}
Due to the centralized statistics estimation of Spark-SQL, execution would
fail for dataset sizes larger than 30GB. However, even for small input, the querying latency of Spark-SQL is $\sim$2
orders of magnitude larger than \sysname's. Druid's 
ingestion would prematurely terminate for dataset sizes $\geq$60GB because the framework (a) indexes data upon ingestion and (b) is optimized for reads over writes~\cite{druidingest}. We did not focus on improving Druid's ingestion. 

\myparaq{Is \sysname accurate and general across \statistics} To evaluate \sysname's accuracy and generality, we look at the accuracy of four different sets containing different numbers of \statistics. Figure~\ref{fig:sec7:caidaAppSet} depicts the boxplot of empirical estimation error for each statistic. Positive error values indicate overestimation errors and  negative error values indicate underestimation. For all application sets, \sysname  operates
under the same resource budget and configuration as described  previously. We find that estimating multiple \statistics does not incur accuracy     reduction, compared to when individual statistics are estimated. \antonis{This highlights  \sysname's
generality, which is enabled by the fact that information maintained in    the
universal sketches is statistic-agnostic and is equally used for multiple statistics of interest.}
\sysname's median estimation error is almost 0 for the L2-norm,  -5.7\% and  -5.5\% for entropy and L1 norm respectively and  9.8\% for cardinality estimation. We can observe that the estimated errors are well within the accuracy threshold that we set. However, for cardinality, we observe a higher median and variance in error values. This is due to a large concentration of $\Gstatistic_i$'s near $\Gstatistic_{min}$. Recall from the discussion of \S\ref{subsec:configs}
that \sysname's error is loosest when $\Gstatistic_i \approx \Gstatistic_{min}$ and this allows for higher error variance.

Figure~\ref{fig:kde} corroborates this observation by depicting the distribution of estimation error values for all summary statistics as a function of the \aggregation's normalized  $\Gstatistic$-sum \ie $\Gstatistic_i/\Gstatistic_{\Datastream}$. Note that for values of $\Gstatistic_i/\Gstatistic_\Datastream\approx\Gstatistic_{min}/\Gstatistic_{\Datastream}$ the variance of empirical error becomes larger as that is the region where the error is allowed to approach our worst-case error bound.  Cardinality estimation using one universal sketch per subpopulation yields estimations with $ <7\%$ error. \antonis{The figure also compares Hydra with uniform sampling and
highlights the high variance in error that sampling exhibits.} We observe the same behavior for the \conviva dataset with a mean error across statistics of $\sim$6\%. 

\begin{figure}[tb]
        \includegraphics[width=1\linewidth]{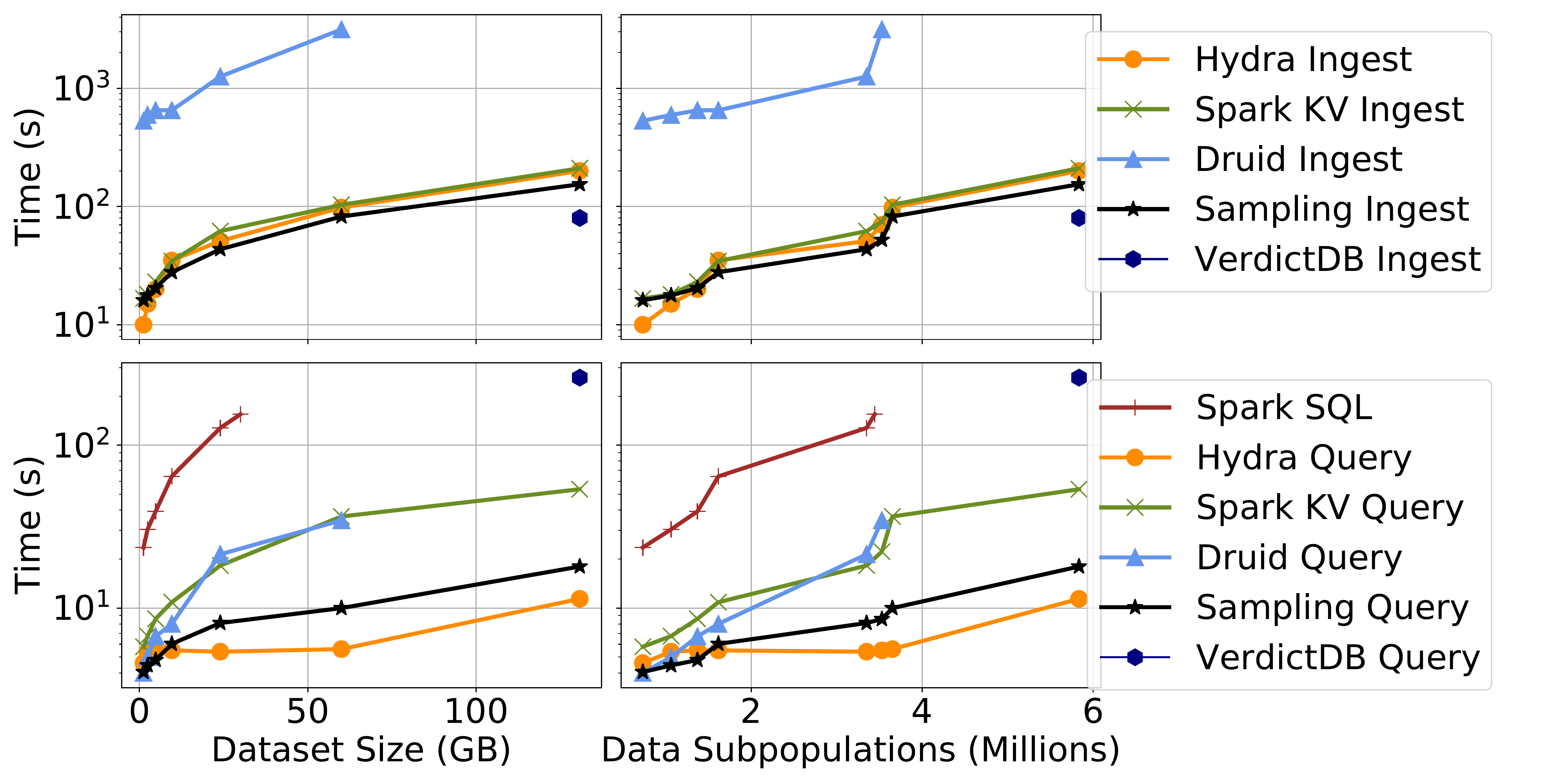}
        \caption{\antonis{Runtime for CAIDA Dataset}}
        \label{fig:sec7:caidaruntime}
\end{figure}

\subsection{Detailed Analysis of \sketch}
\label{subsec:analysis}

First, we compare \sketch's memory footprint to that of our baselines. 
\antonis{Second, we show that our configuration strategies converge to a near-optimal  configuration with respect to memory and runtime. Lastly, we show that
our performance optimizations reduce  \sysname's runtime by 45\%.}

\mypara{Memory Footprint vs. Subpopulations} Figure~\ref{fig:sec7:memoryScalability}
shows memory footprint as a function of the
number of \aggregations monitored for the CAIDA  dataset. \sysname follows the theoretically-expected sub-linear memory scaling as the dataset size and
\aggregations increase.
Indeed, while we observe that for smaller datasets, a Spark-KV   implementation
might be preferable in terms of memory footprint (as the size of the     sketch
instances might even exceed that of the input), this trend is very      quickly
reversed. This is        an
observation that is also confirmed for the \conviva dataset. 

\begin{figure}[thb]
        \includegraphics[width=1\linewidth]{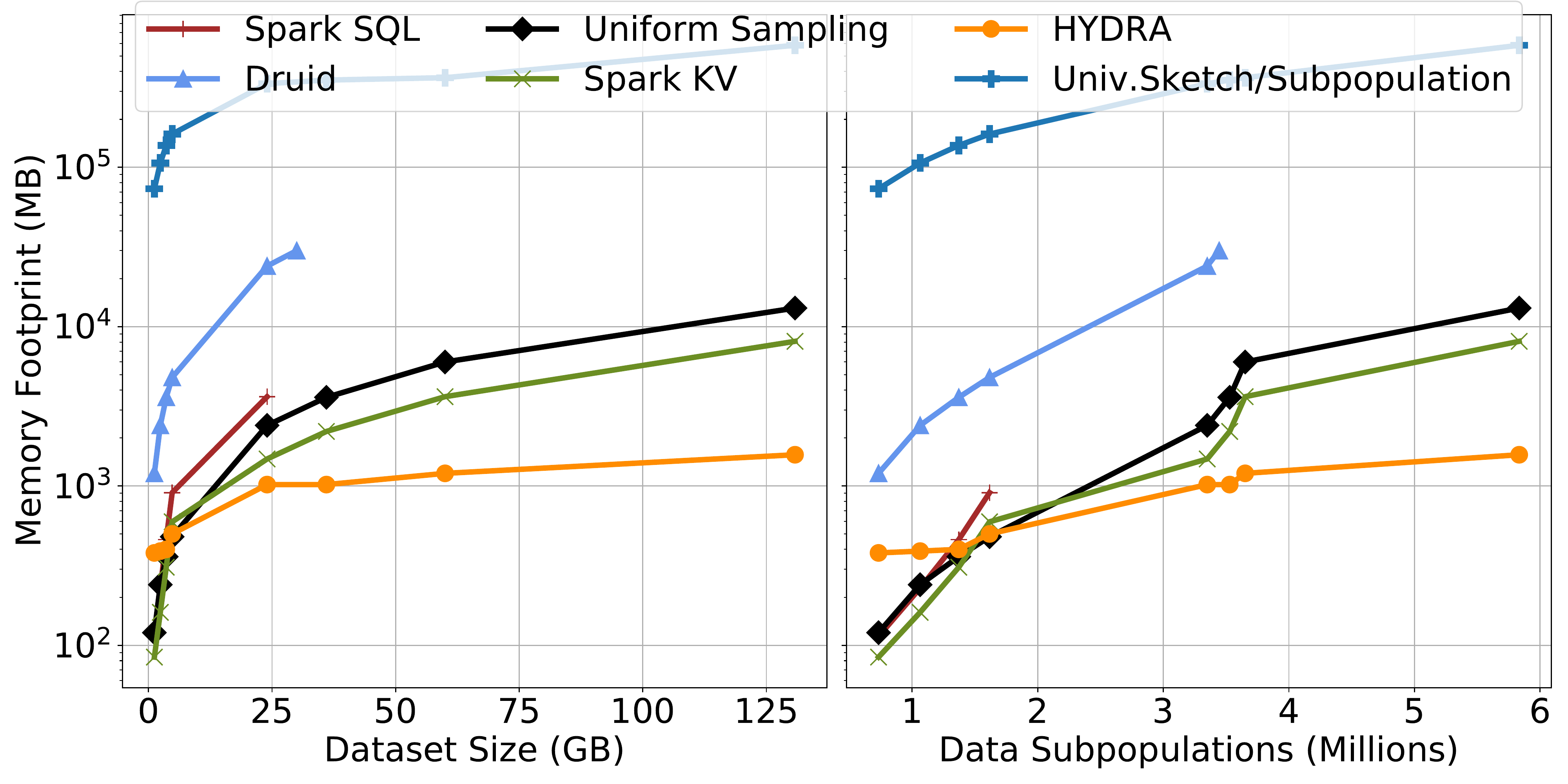}
        \caption{\antonis{Memory footprint per dataset size and subpopulation. VerdictDB numbers do not expose memory utilizations.}}
        \label{fig:sec7:memoryScalability}
\end{figure}

\antonis{\mypara{Configuration Heuristics} Figure~\ref{fig:sec7:Pareto} depicts the 
relationship between the memory footprint of \sketch and its
estimation error for different configurations. The  estimation error    of
the figure is that of the L1-Norm of the       CAIDA
dataset. The optimal configurations 
simultaneously minimize the estimation error and \sketch memory  footprint
(marked with red stars). The orange         diamond
configuration is the suggested configuration based on the    configuration
strategies discussed in \S\ref{sec:sketch}. Thus, our 
strategies result in a configuration comparable to the
optimal configurations. This observation holds across all          summary
statistics and datasets. 

\begin{figure}[h]
        \includegraphics[width=1\linewidth]{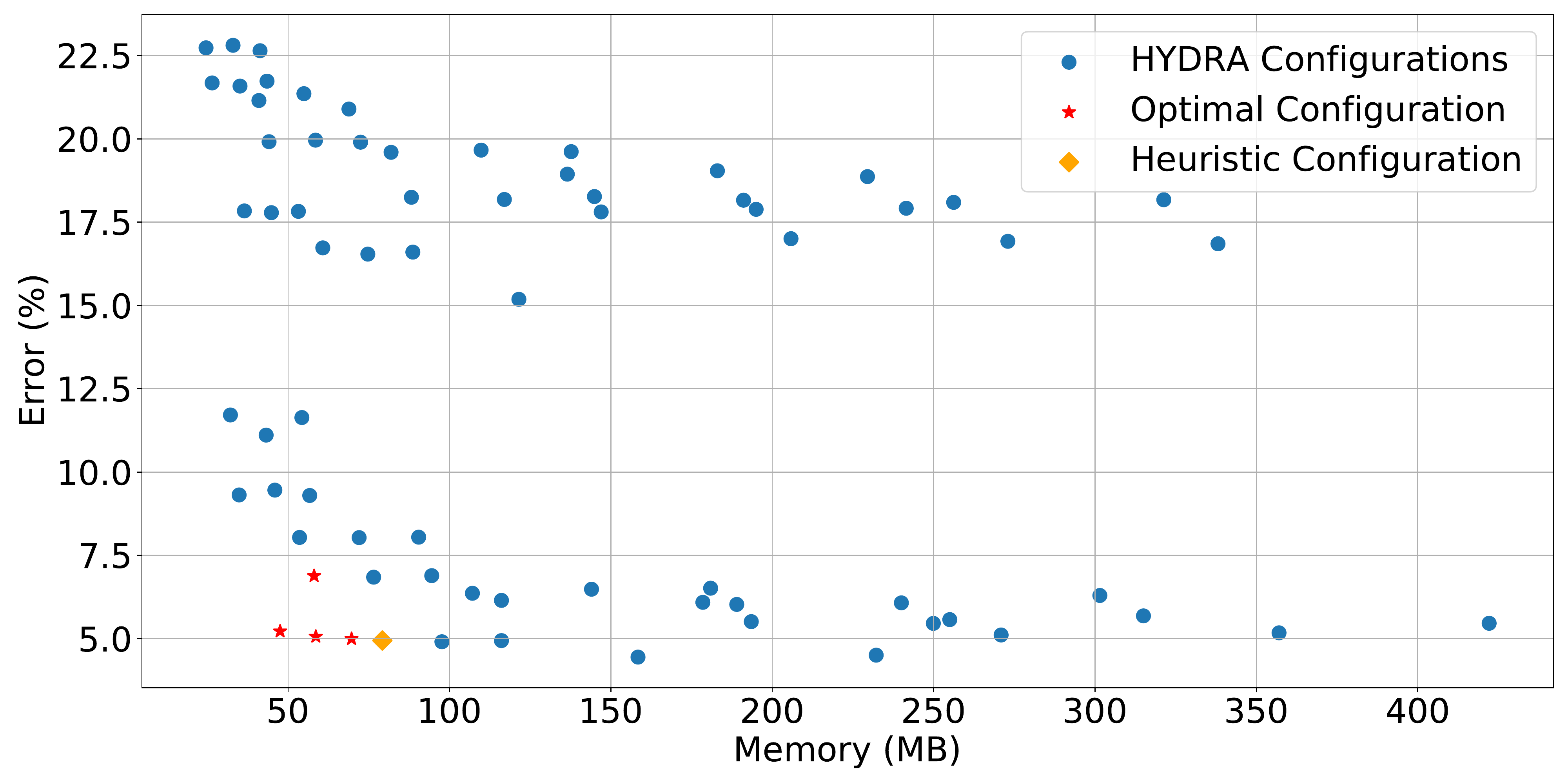} 
        \caption{\sysname's configuration strategies are close to optimal}
        \label{fig:sec7:Pareto}
\end{figure}



\mypara{Analysis of Performance Optimizations} Figure~\ref{fig:sec7:pareto_comparisons}
depicts the cumulative improvement in \sysname's performance using  the
performance optimizations of \S\ref{sec:query}. Each datapoint corresponds
to a different \sketch configuration (the Pareto frontier of 
Figure~\ref{fig:sec7:Pareto}) and we run each configuration twice, once for  the       
basic \sketch design and once with the performance optimizations. The performance
optimizations further reduce the memory footprint of \sketch and the total system runtime. 
Table~\ref{table:optimizations} captures \sysname's runtime reduction
after each performance optimization. The baseline is \sysname without
optimizations; overall, we see a total  performance 
improvement of 45\%. }

\begin{table}[htb]
\centering
\caption{Runtime improvements with performance optimizations}
\begin{small}
    \begin{tabular}{{p{1.6cm}|p{1.6cm}|p{1.6cm}|p{1.6cm}}} 
    \hline
    \textbf{Baseline} & \textbf{Heap-only Merge} & \textbf{One Hash} & \textbf{One Layer Update} \\ 
    \hline \hline
     100\% & 92\% & 64\% & 55\%  \\
    \hline
    \end{tabular}
\end{small}

\label{table:optimizations}
\end{table}

\begin{figure}[thb]
        \includegraphics[width=1\linewidth]{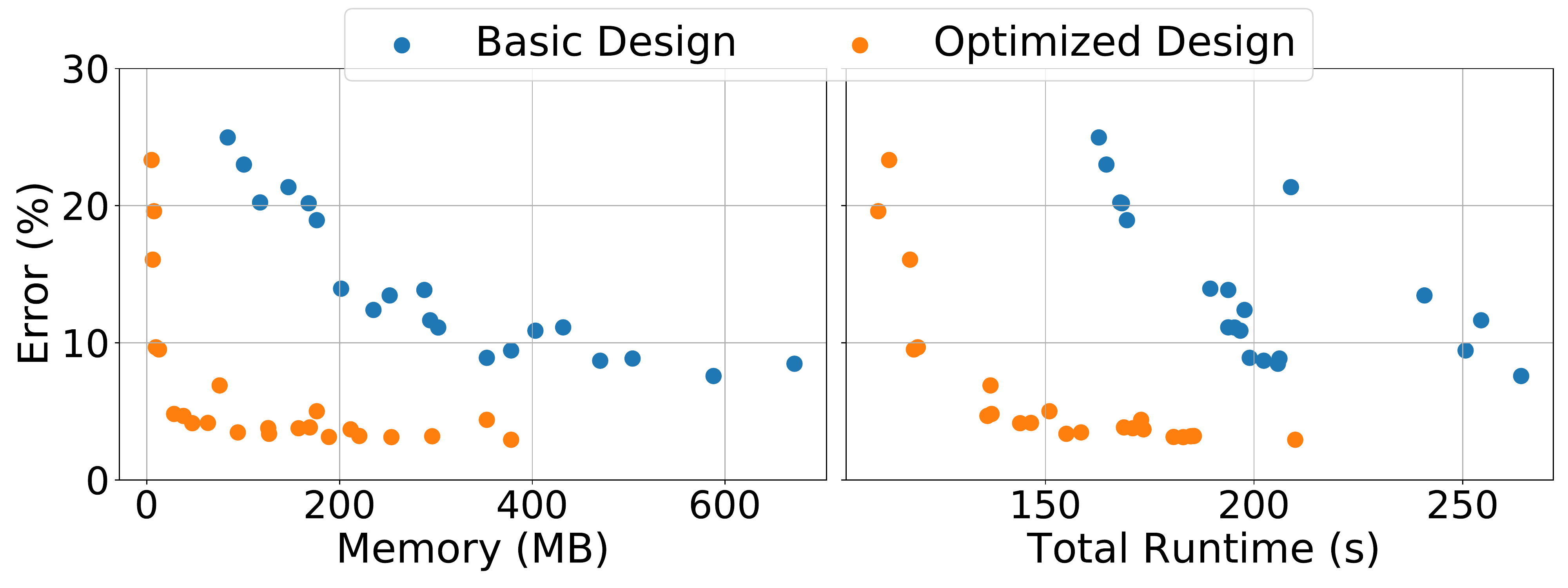}   
        \caption{Comparison of the Pareto frontiers of basic and the optimized
        \sketch implementation for the same configurations.}
        \label{fig:sec7:pareto_comparisons}
\end{figure}

\mypara{Skewness of Dataset} Figure~\ref{fig:sec7:skew}   highlights
the difference in estimation accuracy for two synthetic datasets generated with    a
zipfian distribution. The \aggregations are samples from a zipfian distribution with
parameters $\alpha=0.7$ and $\alpha=0.99$ respectively (a value of  $\alpha=0$
indicates a perfectly uniform distribution). Our experiment confirms our   intuition
that the more skewed dataset ensures a better (memory, error) tradeoff. 
In practice, many real-world datasets are  skewed and thus can benefit from being
analyzed by \sysname. 


\begin{figure}[thb]
        \includegraphics[width=1\linewidth]{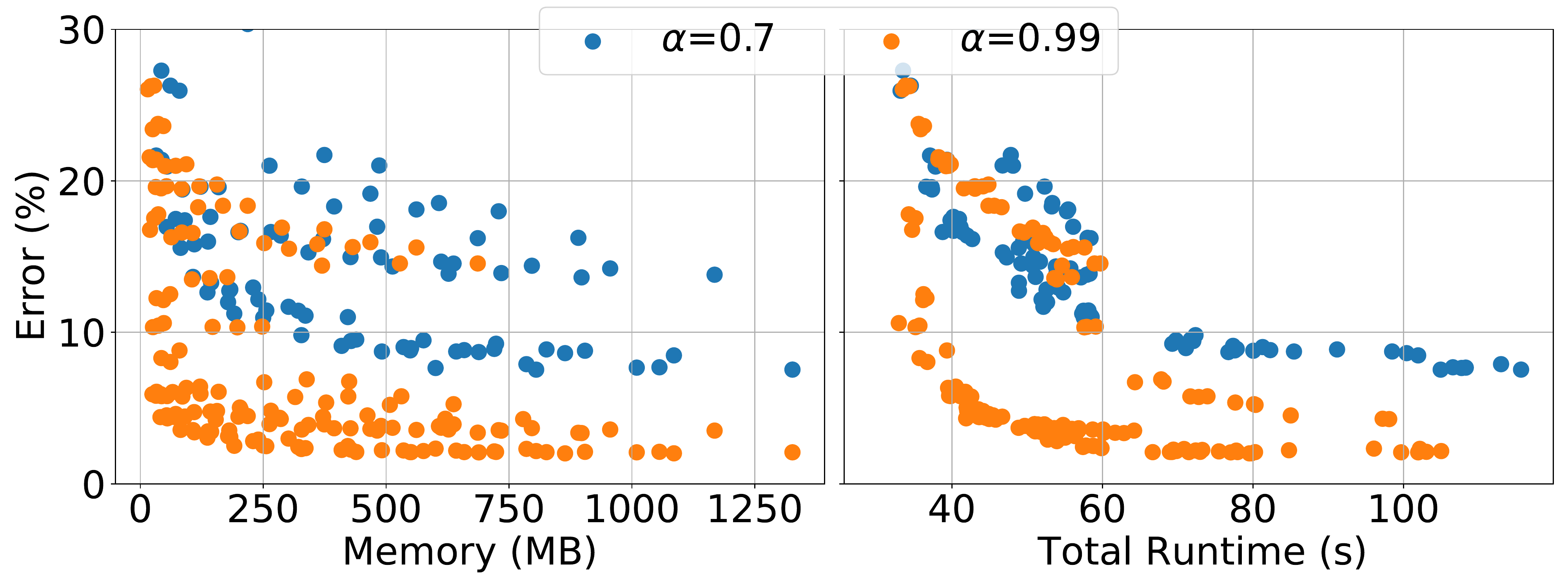} 
        \caption{Impact of data skewness on \sysname's memory footprint and runtime. We use a synthetic dataset where \aggregation sizes are sampled from a Zipfian distribution with parameter $\alpha$.}
        \label{fig:sec7:skew}
\end{figure}

\section{Related work}
\label{sec:related}

\mypara{MapReduce-based Analytics Frameworks} There are various analytics frameworks that are based on the MapReduce paradigm~\cite{hadoop, gfs}. Dryad~\cite{dryad} introduced the concept of user-defined functions in general DAG-based workflows. Apache Drill and Impala~\cite{impala} are limited to SQL variants. Apache Spark~\cite{spark} leverages a DAG-based execution engine and treats  unbounded computation as micro-batches. Apache Flink~\cite{flink} enables pipelined streaming execution for batched and streaming data, offers exactly-one semantics and supports out-of-order processing. \sysname could be built on top of Apache Flink.

\mypara{Stream Processing Frameworks} This line of research focuses on the architecture of stream processing systems, answering questions about out-of-order data management, fault tolerance, high-availability, load management, elasticity \etc~\cite{abadi2005design, arasu2016stream, akidau2013millwheel, kafkastreams,balazinska2005fault, flink, hwang2005high, naiad, streamin, ibms}. Fragkoulis~\etal analyze the state of the art of stream processing engines~\cite{fragkoulis}.

\antonis{\mypara{High-dimensional Data Cubes} Data cubes have been an integral part of online analytics frameworks and enable pre-computing and storing statistics for \multidimensional aggregates so that queries can be answered on the fly. However, data cubes suffer from the same scalability challenges as \sysname. Prior works have focused on mechanisms to identify the most frequently queried subsets of the data cube and optimize operations that are performed only on a small subset of dimensions \textit{at a time}~\cite{li2004high, gray1997data, han2001efficient, harinarayan1996implementing, lakshmanan2002quotient}.}

\mypara{Data Aggregations} The aggregation-based queries that we discussed in    \S\ref{sec:motivation} appear in multiple streaming data systems~\cite{macrobase, braun, gigascope, jetstream, druid, blinkdb, powerdrill} that motivate \sysname. Many of the above frameworks enable approximate analytics but do not fully  satisfy operators' requirements as outlined in \S\ref{sec:motivation}.     

\mypara{Sampling-based Approaches} Multiple analytics frameworks use sampling to provide
approximate estimations~\cite{aqua, olston, Chandola, wang2015spatial}. BlinkDB~\cite{blinkdb} builds stratified samples on its input to reduce query execution
time given specific storage budgets.  STRAT~\cite{strat}  also uses stratified sampling but instead builds a single sample.
SciBORQ~\cite{sidirourgos} builds biased samples based on past query results but cannot provide
accuracy guarantees.

\mypara{Online Aggregation} Online Aggregation frameworks~\cite{hellerstein, pansare,li2016wander} 
continuously refine approximate answers at runtime. In these frameworks, it
is up to the user to determine when the acceptable level of accuracy is reached and to   terminate
estimation. Naturally, this approach is unsuitable for multidimensional telemetry that needs    to
estimate multiple statistics across data subpopulations.

\mypara{Data Summaries} Data ``synopses'' (\eg wavelets, histograms, sketches, \etc) have been extensively
used for data analytics~\cite{agarwal, buragohain, count-min, greenwald, vitter, univmon,wei2015persistent,jestes2011building}. These data summaries can either be lossless or lossy and they aim to provide  efficiency for multidimensional analytics. However, these approaches  are tailored  to a narrow set of estimation tasks.  Gan~\etal develop a compact and efficiently mergeable quantile sketch for multidimensional data~\cite{moment}. 

\antonis{Several prior efforts explore nested sketches as a solution to the multidimensional distinct counting problem~\cite{xiao2015hyper, considine, ting, ting2}. The CountMin Flajolet-Martin (CM-FM) replaces each integer counter of count-min sketch with a distinct counting sketch~\cite{considine}. The CM-FM, while making a step in the right direction for multidimensional analytics, is limited both in terms of the generality and accuracy guarantees it offers~\cite{ting2}. Prior work by Ting~\etal also targets on cardinality estimation in multidimensional data~\cite{ting, ting2} but 
focuses on improving the sketch error bounds. Similar to \sysname, they observe that in distinct counting sketches, accuracy guarantees depend on the characteristics of the underlying data. Their key observation is that the distribution of errors in each counter can be empirically estimated from the sketch itself. By first estimating this distribution, count estimation becomes a statistical estimation and inference problem with a known error distribution. However, computing such error distributions, is computationally heavy in streaming settings as it involves computing maximum likelihood estimators.}

\section{Discussion and Future Work}

\sysname ensures coverage across \aggregations and accuracy guarantees with good resource utilization for subpopulations whose $\Gstatistic_i >= \Gstatistic_{min}$. It is up to the operator to determine $\Gstatistic_{min}$. We believe that this is more versatile than pre-selecting specific \aggregations for which accuracy guarantees should apply. Given a $\Gstatistic_{min}$ threshold, \sysname self-selects the subset of important subpopulations. 

\sysname opens up avenues for future work. For example, an open question is how to enable dynamic sketch reconfiguration given changing workloads or operator goals. Also, a more system-oriented avenue would involve investigating the applicability of \sysname in the context of in-band network telemetry as part of programmable network elements~\cite{sketchlib}.

\section{Conclusions}

Today's large-scale services  require interactive estimates 
of different statistics across subpopulations of their multidimensional
datasets. However, the combinatorial explosion of \aggregations complicates offering multidimensional analytics at a reasonable cost to the operator. We propose \sysname,
a sketch-based framework that leverages \sketch
to summarize data streams in sub-linear memory to the number of \aggregations. We 
show that \sysname is an order of magnitude more efficient than existing analytics
engines while ensuring interactive estimation times.

\begin{acks}
We thank the reviewers for their feedback. This work was supported in part by the CONIX Research Center, one of six centers in JUMP, a Semiconductor Research Corporation (SRC) program sponsored by DARPA, ERDF Project AIDA (POCI-01-0247-FEDER-045907), and NSF awards No. CNS-1513764, CNS-1565343, CNS-2106946, CNS-2107086, SaTC-2132643, and the Kavcic-Moura research award. The authors would also like to thank the Red Hat Collaboratory at Boston University for their support.
\end{acks}

\clearpage
\balance
\bibliographystyle{ACM-Reference-Format}
\bibliography{capp.bib}

\end{document}